\documentclass{article}

\usepackage[english]{babel}

\usepackage[letterpaper,top=2cm,bottom=2cm,left=3cm,right=3cm,marginparwidth=1.75cm]{geometry}

\usepackage[utf8]{inputenc}
\usepackage[T1]{fontenc}
\usepackage{microtype}

\usepackage{amsmath,amsthm,amsfonts}
\usepackage{graphicx}
\usepackage[colorlinks=true, allcolors=blue]{hyperref}
\usepackage{natbib}
\usepackage{tikz}
\usepackage{pgfplots}
\usepackage{xurl}

\usepackage{subcaption}
\usepackage{xcolor}

\usepackage{mathtools}
\usepackage{algorithm}
\usepackage{algorithmicx}
\usepackage{algpseudocode}
\makeatletter
\renewcommand\fs@ruled{%
  \def\@fs@cfont{\bfseries}%
  \let\@fs@capt\floatc@ruled
  \def\@fs@pre{{\color{black}\hrule height.8pt depth0pt}\kern2pt}%
  \def\@fs@post{\kern2pt{\color{black}\hrule}\relax}%
  \def\@fs@mid{\kern2pt{\color{black}\hrule}\kern2pt}%
  \let\@fs@iftopcapt\iftrue
}
\makeatother

\floatstyle{ruled}
\restylefloat{algorithm}

\usepackage{enumitem}

\usepackage{enumitem}
\usepackage{subcaption}

\theoremstyle{definition}
\newtheorem{lemma}{Lemma}
\newtheorem{theorem}{Theorem}
\newtheorem{example}{Example}

\newtheorem{proposition}{Proposition}

\newcommand{\diag}{{\mathrm{diag}}}

\renewcommand{\deg}{{\mathrm{deg}}}

\newcommand{\xv}{{\pmb{x}}}

\newcommand{\ev}{{\pmb{e}}}
\newcommand{\fv}{{\pmb{f}}}

\newcommand{\yv}{{\pmb{y}}}

\newcommand{\Da}{{\Lambda_{\alpha}}} 
\newcommand{\Db}{{\Lambda_{\beta}}} 

\definecolor{myblue}{RGB}{90,128,184}
\definecolor{myorange}{RGB}{216,127,64}
\definecolor{mygreen}{RGB}{126,147,74}

\definecolor{DarkGreen}{rgb}{0.075,0.375,0.075}
\definecolor{DarkRed}{rgb}{0.5,0.1,0.1}
\definecolor{DarkBlue}{rgb}{0.1,0.1,0.5}
\definecolor{Gray}{rgb}{0.2,0.2,0.2}

\hypersetup{
	colorlinks=true,       
	linkcolor=DarkBlue,          
	citecolor=DarkGreen,        
	filecolor=DarkBlue,      
	urlcolor=DarkBlue,          
	pdftitle={},
	pdfauthor={},
}

\title{Reaching a Consensus in Predictive Loops}

\usepackage{authblk}
\stepcounter{footnote}
\addtocounter{footnote}{-1}
\author[1,2]{Jiduan Wu}
\author[1,3]{Rediet Abebe$^*$}
\author[1,3]{Celestine Mendler-D\"unner\footnote{joint supervision.}}

\affil[1]{Max Planck Institute for Intelligent Systems, T\"ubingen, and T\"ubingen AI Center}
\affil[2]{ETH Zurich}
\affil[3]{ELLIS Institute Tübingen}

\date{}        
\begin{document}

\vspace{0.5cm}
\maketitle
\vspace{-1cm}
\begin{abstract}

Predictions in digital platforms must adapt over time as individuals update their beliefs through social interactions. At the same time, changing predictions alter the content people are exposed to and, consequently, the very beliefs they aim to forecast. 
This recursive coupling between predictions and individuals complicates the analysis of the long-term societal impact of predictive systems.
In this work, we propose a minimal model where predictions and opinions co-evolve, combining insights from network science with concepts from performative prediction.
In our model a platform's predictions influence individual opinions, which then evolve through peer interactions and form the training data for future platform model updates. We demonstrate that this co-evolution induces a novel equilibrium that qualitatively differs from standard network equilibria. In particular, we show how standard predictive objectives can drive networks toward consensus even under conditions where classical opinion-dynamics models lead to disagreement. This emerges because predictive systems dynamically adapt to changing opinions, and learning objectives create spillover effects among individuals beyond the topology of the network. We further analyze systematic deviations from standard prediction and demonstrate amplified effects of targeted platform interventions on equilibrium outcomes, compared to classical network intervention analyses. Together, our results illustrate performativity as an important, yet so far neglected, qualifying factor in social networks.~\looseness=-1

\end{abstract}

\begin{figure*}[h!]
\vspace{0.5cm}
\centering
\includegraphics[width=0.8\textwidth]{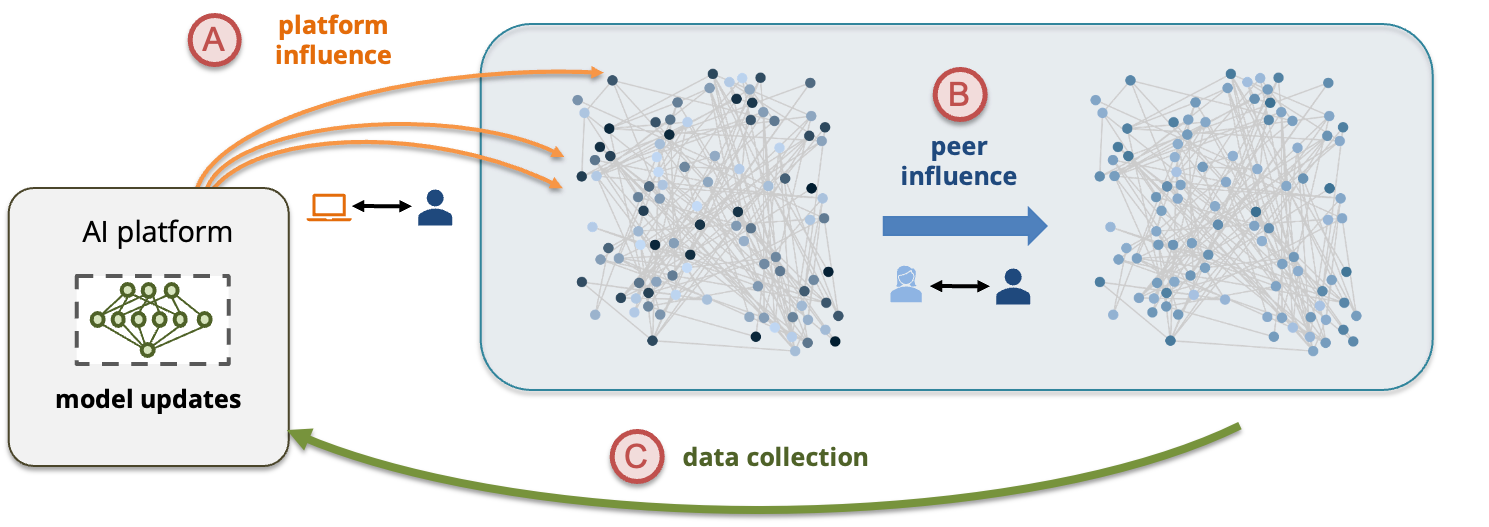}
\caption{\emph{Peer-platform co-influence in opinion dynamics}:~(A)  Individuals interact with a digital platform that deploys a predictive system and update their opinions based on the content they are exposed to.~(B) Individuals' opinions evolve through peer interaction.~(C)  Expressed opinions inform future platform predictions.}
\label{fig:teaser}
\end{figure*}

\section{Introduction}

The formation of opinions has traditionally been modeled as a social process unfolding through interpersonal influence: individuals exchange views with their peers, update their beliefs, and gradually converge or polarize depending on the structure of the network and the strength of social ties. From classical models of diffusion~\citep{degroot1974reaching, holley1975ergodic, friedkin1990social} to contemporary theories of polarization~\citep{bail18polarize}, opinion dynamics have focused on how local interactions between peers aggregate into collective outcomes. However, this assumption no longer accurately reflects how opinions form in today's society where digital platforms play an increasingly dominant role as a source of information.
~\looseness=-1

The effect of digital platforms on beliefs and opinions has been documented in various settings, ranging from preferences for restaurants~\citep{anderson12yelp} and travel providers~\citep{ursu18expedia} to political leaning~\citep{gauthier2026political, huszar2022algorithmic} and medical choices~\citep{curry2014prevalence,farnan2013online}. To account for such platform influence on the formation of opinions several extensions of classical opinion dynamics models have been proposed. Among others, they discuss how platforms change beliefs, modify social ties, and mediate trust. However, these approaches predominantly treat platforms as static or exogenous to the opinion formation process~\citep{out2024impact}. As a result, they miss how digital platforms respond to changing opinions and repeatedly distil aggregate information into future decisions.

To address this, the core focus of this work lies in the interaction of individuals with machine-learning-powered  systems. These are algorithms that are continuously fed with behavioral data about users, make predictions about user preferences and determine what content to show next. As argued above, even when framed as purely predictive, such algorithmic decisions shape downstream user behavior and future data observed by the algorithm---a phenomenon known as \emph{performativity}. As a consequence, learning systems are not merely descriptive or static; rather, they co-evolve with the data they create. This dynamic is conceptualized in the machine learning literature under the umbrella of performative prediction~\citep{PZMH20, hardt25sts}.
Applied to opinion dynamics, performativity implies that any  platform applying predictive algorithms, necessarily becomes an active participant in the opinion formation process. 
This socio-technical feedback-loop is illustrated in Figure~\ref{fig:teaser}. Predictions and opinions are recursively coupled: observations of opinions determine predictions and the performative effects of predictions determine opinions. We refer to this feedback mechanism, characteristic of machine learning powered platforms, as the predictive loop.

To understand how predictive loops alter outcomes of collective opinion formation, we blend models of social influence~\citep{friedkin1990social} with concepts from performative prediction~\citep{PZMH20}. With a minimal model for how predictions and opinions co-evolve, we illustrate how platforms can influence collective outcomes over time, spread information across subgroups, and alter collective outcomes in ways that neither framework alone is able to capture.

\subsection{Our work}

We model the co-evolution of opinions and predictions as a nested feedback loop. The inner-loop describes how opinions evolve through peer interactions, and the outer-loop describes how platforms influence the opinion formation process and subsequently learn from observations to determine future predictions. An equilibrium state emerges where the platform has no reason to change its predictions based on observed data. This new equilibrium state is not only determined by network characteristics and individuals' susceptibility to peer influence, as in traditional network models, but also the platform's choice of learning algorithm and individuals' susceptibility to platform influence. With technical results and simulations, we perform investigations into the characteristics of this new equilibrium and illustrate interesting deviations from the classical network equilibria absent performativity. 

More concretely, we instantiate peer dynamics with the classical Friedkin-Johnsen model~\citep{friedkin1990social}, and focus on self-fulfilling predictions on the side of the platform. Through a theoretical analysis of non-parametric prediction and extensive simulations using parametric models our work offers the following  key take-aways:

\begin{itemize}[topsep=0pt, itemsep=0pt]
    \item \emph{Predictive loops as an explanation for consensus.} We show that by echoing expressed opinions, performative predictions can act as a homogenizing force in networks. For our model we prove that under natural assumptions on the network, the variance of opinions at equilibrium is negatively correlated with the susceptibility of individuals to platform influence, or, in other words, the strength of performativity. Notably, in the limit of high platform susceptibility, a consensus is reached. This consensus emerges beyond primitive graphs and despite stubbornness during peer interactions---standard conditions under which classical results in networks would suggest disagreement.~\looseness=-1 
     \item \emph{Performativity as a pathway for peer influence.} 
     We illustrate how extrapolating predictions to unobserved individuals can induce consensus under even weaker conditions on the network. In particular, we prove that it is sufficient for information to flow in the observed subgroup for a consensus to emerge across the entire population, again, in the limit of high platform susceptibility. The reason is that individuals who are not susceptible to peer influence can indirectly be impacted by the opinions of their peers through model predictions. We offer a lower bound on this performativity-mediated spillover effect unique to predictive loops. Similar in vein, we show how interventions targeted at a single individual impact equilibrium opinions across the entire population in a way that is correlated with the strength of performativity.
     \end{itemize}
Both of these take-aways constitutes an important conceptual deviation from standard theories of networks and showcase performativity as an important qualifying factor in social dynamics.

\section{Preliminaries and related work}
\label{sec:preliminaries}

Throughout this work, graphs serve as a mathematical object to characterize interpersonal relationships in a population of individuals. Each node in the graph corresponds to an individual, and an edge between two nodes represents a social connection. Individuals then interact with a digital platform external to the network.
In the following, we introduce concepts and notation, and provide the necessary background on the individual parts composing the predictive loop as illustrated in Figure~\ref{fig:teaser}.

\vspace{-0.3cm}

\paragraph{The social network graph.}
We use $\mathcal{G}=(\mathcal{V},\mathcal{E})$ to denote a undirected simple graph where $\mathcal{V}$ with $\mathcal{V}=[n]$ is the set of nodes and $\mathcal{E}=\{(i,j)\mid\text{$i$ and $j$ are connected}\}$ represents the set of edges. Assume each node $i\in \mathcal{V}$ represents an individual on the social network $\mathcal{G}$. The adjacency matrix of $\mathcal{G}$ is denoted as $A$ where $a_{ij}=1$ if $(i,j)\in \mathcal{E}$, and $a_{ij}=0$ otherwise. We denote the degree of each node as $d_i$ and use $\deg\coloneqq\diag(d_1,\cdots,d_n)$. 
Each individual $i$ on the network holds a time-dependent opinion~$x_i$. We use the vector notation $\pmb x=[x_1,...,x_n]$ to denote opinions of all $n$ individuals in compact form. 
These opinions are our main subject of study and they can be influenced by both peer interactions and platform predictions. If not specified otherwise we assume the graph  $\mathcal G$ underlying the population is connected.~\looseness=-1

\subsection{Peer influence}
\label{sec:prelimPeer}

An individual's opinion can be influenced by social interactions. This is the basis for the long-standing field of opinion dynamics that seeks to model how beliefs evolve on a social network graph. Early work by \citet{degroot1974reaching} describes belief evolution as a repeated linear averaging mechanism among neighbors in a graph and serves as a baseline model of consensus formation. Later, the influential Friedkin–Johnsen model \citep{friedkin1990social, friedkin1999social} anchors individuals' opinions in their initial beliefs, allowing persistent disagreement even in connected networks. There are numerous extensions and variations of these base models that consider stubbornness of individuals~\citep{yildiz2013binary}, bounded confidence models~\citep{rainer2002opinion}, as well as alternative mechanisms for the aggregation of opinions of peers~\citep{mei22median}. We refer interested readers to recent reviews of opinions dynamics in social networks~\citep{proskurnikov2017tutorial,peralta2025opinion}. In our work, we build on models of opinion dynamics to describe how platform-influenced opinions spread across a population in between individual retraining rounds.

More specifically, we build on the Friedkin-Johnsen (FJ) model~\citep{friedkin1990social}, describing the iterative process of opinion formation across a graph with $n$ nodes. Let the vector $\xv^{(k)}\in[0,1]^n$ denotes individuals' opinions at time step~$k$, with~$\xv^{(0)}=\xv_\mathrm{init}\in[0,1]^n$ referred to as initial opinions. Then, the FJ model is defined by the following recursion   
\begin{equation}
\xv^{(k+1)}=(I_n-\Lambda_{\alpha})\,\xv_\mathrm{init}+\Lambda_{\alpha} W\xv^{(k)}.\label{eq:FJ}
\end{equation}
The influence matrix $W$ defines how individuals weight the opinions of their peers during social interaction. For the purpose of this work the influence matrix is defined by the graph structure as $W:=\mathrm{deg}^{-1}A$ where $A$ denotes the adjacency matrix, hence, it is constant and row-stochastic. The second parameter governing the dynamics is the peer influence matrix $\Da=\diag(\alpha_1,\cdots,\alpha_n)$ with $\alpha_i\in[0,1]$ for $i\in[n]$ parameterizing individual $i$'s susceptibility to peer influence. For $\alpha_i=0$, individual $i$ does not listen to their peers and maintains their initial opinion. As $\alpha_i$ increases, individual $i$ becomes more susceptible to peer influence. For~$\alpha_i\equiv1$,~$\forall i\in[n]$, the dynamics model in Equation~\ref{eq:FJ} recovers the linear averaging mechanism of the DeGroot model~\citep{degroot1974reaching} as a special case. The complement $1-\alpha_i$ is commonly referred to as the stubbornness parameter. 

\vspace{-0.3cm}

\paragraph{Standard network equilibria.}
We use $\mathrm{FJ}_k(\xv_\mathrm{init})$ to denote the state of opinions~$\xv^{(k)}$ after~$k$ discrete time steps of FJ dynamics when initialized at~$\xv_\mathrm{init}$. Similarly, we denote the equilibrium state for~$k\rightarrow\infty$ by $\mathrm{FJ}(\xv_\mathrm{init})$. When~$\alpha_i<1$ for at least one individual~$i$, the dynamics are guaranteed to converge when $\mathcal G$ is connected. The equilibrium can be characterized in closed form:~$\mathrm{FJ}(\xv_\mathrm{init})=(I_n-\Da W)^{-1}(I_n-\Da)\xv_\mathrm{init}$~\citep{parsegov2016novel}. For the special case of DeGroot, the equilibrium corresponds to  $\mathrm{DG}(\xv_\mathrm{init})=\lim_{t\rightarrow \infty} W^t\; \xv_\mathrm{init}$ while additional assumptions are required to guarantee convergence for $\alpha_i\equiv1$, typically expressed as a primitivity assumption on $\mathcal G$~\citep[see, e.g.,][]{degroot1974reaching}.  We will use these standard  network equilibria as a baseline  for our analysis.

\subsection{Platform influence}
\label{sec:prelimPP}

It has been widely documented that beyond peer interaction, the exposure to online content on digital platforms shapes individuals' opinions. For example, yelp recommendations shape restaurant choices~\citep{anderson12yelp}, TV advertisements influences consumer preferences~\citep{adomavicius2013recommender}, 
and ranking algorithms shape clicks~\citep{narayana15click,ursu18expedia,mendler2024engine}. Phenomena such as the spreading of misinformation~\citep{del2016spreading}, filter bubble~\citep{pariser2011filter}, polarisation~\citep{perra2019modelling},  and echo chamber~\citep{cinelli2021echo,chaney2018algorithmic} can be unwanted outcomes from it.

Prior work in opinion dynamics accounts for platform influence by modeling how platforms intervene on the network structure. For example, \citet{zhou2024modeling} and \cite{chitra2019understanding} model the influence of platform recommendations as changes to the edge weights.  \citet{santos21link} and \citet{wang2023relationship} study how link recommendations can influence network structure and equilibria. \citet{out2024impact} study the influence of biased media by modeling the platform as a node with its own expressed opinion, connected to every individual in the network.
\cite{pescetelli2022bots} treat platforms as a filter that ranks the opinions of neighbors and study the equilibrium effect of bots using agent-based simulation.  
In our model, platforms directly influence individuals' opinions through self-fulfilling predictions. Further, the platform is external to the graph and offers personalized predictions to every individual.

Akin to \citet{jia2015opinion,jia2019opinion} and \citet{friedkin2016theory}, we consider the evolution of opinions across multiple stages of peer interactions. But instead of reasoning about changes in the weighting of individuals through social dynamics, we leave the network untouched and study the evolution of opinions when the platform updates the prediction for each individual before each round of peer interactions.

\subsection{Performative prediction} 

Predictive models are rarely static but repeatedly updated based on fresh data. If this data is influenced by the model itself it creates a dynamic in which predictions and data co-evolve. In the machine learning literature such learning dynamics have been formalized under the umbrella of performative prediction~\citep{PZMH20}. Convergence  guarantees and equilibria are derived under different learning settings \citep[e.g.,][]{PZMH20, mendler20stochasticPP,drusvyatskiy23stochastic,wang2023constrained,mehrnaz23pp, farina2026stability}. While these results characterize the asymptotic bahavior of the platform, they model the population as a distribution without specifying the internal structure. Instead, in this work we explicitly leverage network structure present in the population and focus on equilibrium outcomes for the population, rather than properties of the learned model.~\looseness=-1 

Other works in performative prediction that study explicit models for the data generating process, such as rational agent models~\citep{hardt16strategic,jag21alt}, causal models~\citep{mendler22causal, cheng2023causal}, or parametric models~\citep{miller2021outside,izzo21gd}, have not considered network effects among individuals. \citet{wang2023network,narang23multi,saig2025evolutionary} consider network effects among decision makers. \citet{eilat2023performative} discuss a bipartite graph between strategic users and recommended items. But none of these existing models aims to characterize equilibrium outcomes for individuals in networked populations, not does this derive from existing results. We refer to \citet{hardt25sts} for a more comprehensive overview on the literature of performative prediction. For the rest of this work, the most important concept we build on is the notion of performative stability, characterizing a fixed point of model retraining, and, hence, an equilibrium concept we are interested in. We will explain this more in the next section.~\looseness=-1

\vspace{-0.3cm}

\paragraph{Mathematical conventions.}
Before presenting our model it remains to set up necessary notation. We let $(A)_{ij}$ denote the entry at $i$-th row $j$-th column of matrix $A$. The transpose of a matrix $A$ or a vector~$\pmb{v}$ are denoted by $A^{\top}$, $\pmb{v}^{\top}$ respectively. An identity matrix of dimension $d$ is denoted by $I_d$. Let $\ev_i$ be the $i$-th standard basis vector whose $i$-th entry is 1 and remaining entries are $0$. We say that a real matrix is non-negative (positive) if all entries $(A)_{ij}$ are non-negative (positive).  We write~$A=\diag(a_1,\ldots,a_d)$ for a diagonal matrix with entries $(A)_{ij}=a_i$ for $i=j$ and $0$ otherwise.  We denote $A\leq B$ if matrices $A,B$ with same dimensions satisfying $(A)_{ij}\leq(B)_{ij}$ $\forall i,j$. The same convention is used for vectors. A non-negative square matrix $A\in\mathbb{R}^{n\times n}$ is irreducible if $\forall i,j\in[n]$, there exists $m_{ij}\in\mathbb{N}$ such that $(A^{m_{ij}})_{ij}>0$. We use $\sigma(A)$ to denote the spectrum of the matrix $A$ and~$\rho(A)$ to denote the spectral radius of $A$. Moreover, $A$ is primitive if~$\exists m\in\mathbb{N}$ such that $(A^m)_{ij}>0$,~$\forall i,j$. A non-negative matrix $A$ is row-stochastic if $\sum_j A_{i,j}=1$, $\forall i$. We write $a = \mathcal{O}(b)$ if there exist
constants~$C > 0$ such that $a \le C b$.

\begin{algorithm}[t!]
\caption{Algorithmically-mediated opinion dynamics}
\label{alg:retrainMain}
\begin{algorithmic}[1]
\renewcommand{\algorithmicrequire}{\textbf{Parameters:}}                     
\Require ${\pmb{\alpha}}, \pmb{\beta}, K$, a connected simple graph $\mathcal{G}$
\renewcommand{\algorithmicrequire}{\textbf{Input:}}
\Require Innate opinion $\xv^*$, platform algorithm $\mathrm{Algo}(\cdot)$.
\State $\fv^{(1)} = \mathrm{Algo}(\xv^*)$.
\For{$t>0$}  
\State {$\pmb{x}_\mathrm{init}^{(t)}=(I_n-\Db)\pmb{x}^*+\Db \pmb{f}^{(t)}$\Comment{\colorbox{myorange!30}{(A) platform influence}}.}
    \State  {$\xv_0=\xv_\mathrm{init}^{(t)}$.}
    \For {$k=0,1,2,\dots,K-1$}
        \State {$
    \pmb{x}_{k+1}=(I_n-\Da) \pmb{x}_0+ \Da W\pmb{x}_{k}.$ \Comment{\colorbox{myblue!30}{(B) peer influence}}}
    \EndFor
    \State {Platform observes expressed opinions $\xv_{\mathrm{ex}}^{(t)}=\xv_K$.}
    \State {$\pmb{f}^{(t+1)}=\mathrm{Algo}(\xv_\mathrm{ex}^{(t)}).$\Comment{\colorbox{mygreen!30}{(C) platform model update}}}
\EndFor
\end{algorithmic}
\end{algorithm}

\section{A model of peer-platform co-influence}
\label{sec:model}

We start by setting up our model to describe evolving opinions subject to both peer and platform influence. We consider a population of $n$ networked individuals repeatedly interacting with a platform. During each interaction the platform recommends potentially personalized content which influences  individuals opinions. Subsequently, individuals share opinions with their peers and influence each other. After a fixed time period, the platform observes expressed opinions and adjusts future predictions accordingly. We treat each period of peer interaction as a separate instances of the FJ model where initial opinions depend on the platform's active predictions. Subsequent instances are coupled in that predictions of the platform at instance $t$ come from a predictive model trained on the expressed opinion of individuals at instance $t-1$. We are interested in the long-term equilibrium of this coupled dynamic.

Our concrete model is summarized in Algorithm~\ref{alg:retrainMain}. We use $t=1,2,3,...$ to index platform interactions and separate instances of the FJ model. Accordingly, we use $\xv_\mathrm{init}^{(t)}$ to denote the initial opinion of individuals in instance $t$, and $\fv^{(t)}\in [0,1]^n$ to encode predictions for the $n$ individuals during instance $t$. The vector $\xv^*$ denotes individuals' innate opinions, which can be thought of as prior beliefs, personal values or prejudices. To model platform influence, we assume initial opinions are a biased version of individuals' innate opinions, influenced by the predictions of the platform. The strength of the bias depends on individuals' susceptibility to platform influence. More formally,
\begin{equation}
\xv_\mathrm{init}^{(t)}=(I_n-\Db)\xv^*+\Db \fv^{(t)},
\label{eq:P0}
\end{equation}
where the diagonal matrix $\Db=\diag(\beta_1,\ldots,\beta_n)$ with $\beta_i\in[0,1]$ denotes the susceptibility of each individual $i\in[n]$ to platform influence. For $\beta_i=0$, the initial state of individual $i$ corresponds to their innate opinions $x_i^*$, akin to the classical interpretation of the FJ model. The larger $\beta_i$, the more individual $i$'s opinion is confounded by the predictions. 

Then, given the initial opinion $\xv_\mathrm{init}^{(t)}$, we assume peer dynamics follow the FJ model, as detailed in Equation~\ref{eq:FJ}. This means, individuals repeatedly adjust their opinion based on the opinion expressed by their neighbors in the graph, where $\alpha_i$ encodes individual $i$'s susceptibility to peer influence. For $\alpha_i=0$, individual $i$ does not change their opinion during peer interactions. The number of peer interaction steps $K$ performed in each instance is a parameter of our model. For $K=\infty$ we let the FJ dynamics evolve until an equilibrium is reached.

The digital platform then observes the resulting expressed opinions. For each instance $t$ we denote this observed state of opinions after $K$ interaction steps as
\[\xv_\mathrm{ex}^{(t)}:=\mathrm{FJ}_K(\xv_\mathrm{init}^{(t)}).\]
Observations of $\xv_\mathrm{ex}^{(t)}$ are then collected into a dataset and form the basis for the platform to retrain their predictive model. Thus, observations from time step $t$ determine predictions deployed by the platform at time step $t+1$. To indicate that a learning algorithm is involved in this step, we write
\begin{align}
\pmb{f}^{(t+1)}=\mathrm{Algo}\big(\xv_\mathrm{ex}^{(t)}\big).\label{eq:algo}
\end{align}
The algorithm $\mathrm{Algo}(\cdot)$ maps observations to predictions. It serves a flexible placeholder for the platform's learning algorithm and we will discuss specific instantiations in the following sections. Think of the argument $\xv$ as the variable that the platform observes, which is simultaneously the target of prediction. What is implicit in our notation is that the platform might have access to additional meta-data about individuals to train a parametric prediction model to predict opinions. For our theoretical analysis we assume platforms can fully observe opinions, which constitute the target variable in our prediction task. Practical challenges of estimation and measurement are left for future work. Below we provide an real-world example to motivate our framework.

\begin{example}[Content recommendation]
    Envision a setting in which individuals interact with a recommendation platform while forming opinions on a personal lifestyle choices such as smoking or a consequential public issues such as a legal proposal. Each individual’s opinion is represented as a continuous value on a spectrum between two opposing positions. The platform recommends content, including videos, podcasts, and news articles, with varying positions on the issue based on a predictive algorithm and users’ previously expressed opinions. On a daily basis, individuals consume the recommended content, update their opinions, and exchange views with peers through online sharing or in-person conversations. As this process repeats, the platform continuously adapts its recommendations using newly observed behavioral data.
\end{example}

\subsection{Performative stability}
We are interested in the equilibrium opinions under our new model. This equilibrium  differs from classical network equilibria characterizing converging peer dynamics due to the outer loop: even if peer dynamics equilibrate there might still be updates to the predictive model that cause opinions behave differently in the next instance. Thus, we say the dynamics in Algorithm~\ref{alg:retrainMain} have converged if there is no reason for the platform to update their predictive model based on the data they observe. This means, the algorithm continues to recover the model that is already deployed---and hence it repeatedly recovers the same outcome after peer interaction. Following the terminology of \cite{PZMH20}, we denote this equilibrium state \emph{performative stability}.

We say predictions $\fv_\mathrm{PS}$ are performatively stable iff they satisfy the following fixed point equation
\[\fv_\mathrm{PS}=\mathrm{Algo}\big(\mathrm{FJ}_K((I-\Lambda_\beta)\xv^* + \Lambda_\beta \fv_\mathrm{PS})\big).\]

We use $\xv_\mathrm{PS}$ to denote the equilibrium opinions implied by $\fv_\mathrm{PS}$, i.e., $\xv_\mathrm{PS}=\mathrm{FJ}_K((I-\Lambda_\beta)\xv^* + \Lambda_\beta \fv_\mathrm{PS})$. 
Throughout the paper, we discuss instantiations of our model where Algorithm~\ref{alg:retrainMain} converges to a unique performatively stable point. Accordingly, the equilibrium of interest is
\begin{equation}\lim_{t\rightarrow\infty}\mathrm{FJ}_K\big(\xv_\mathrm{init}^{(t)}\big)\rightarrow \xv_\mathrm{PS}.
\label{eq:PS}
\end{equation}

Going forward, we characterize how $\xv_\mathrm{PS}$ depends on $\beta$, $\alpha$ and $K$ for different algorithmic strategies, and analyze how it differs from  classical network equilibria emerging under the Friedkin-Johnsen model in isolation.
A particularly interesting case of our model is when $K\rightarrow\infty$. Here the platform observes opinions at equilibrium and the performatively stable point corresponds to a state where both, learning dynamics and peer dynamics have converged. However, keeping $K$ general allows for a more expressive algorithm and a more general analysis that can cover different length of the deployment cycle. Throughout,  we refer to the reference point $\xv_\mathrm{FJ}:=\mathrm{FJ}(\xv^*)$ as the platform-free equilibrium.

\section{Peer dynamics when perfectly predicting the past}
\label{sec:supervised_learning}

Predictions rely on observations of expressed opinions from the past. At the same time they influence opinions the moment they are used to make recommendations. This can couple instances of peer interaction happening at different time steps. To analyze the effects, we focus on perfect prediction. It represents an idealized case of model retraining where the platform can accurately predict the opinion of each individual in the network. Hence, the learning algorithm simplifies to  
\begin{equation}\mathrm{Algo}(\xv_\mathrm{ex}^{(t)}):=\xv_\mathrm{ex}^{(t)},\quad t\geq1.
\label{eq:SL}
\end{equation}
This baseline instantiation exposes important qualitative effects of performative learning systems without getting into the complexities of data collection, model training and uncertainty quantification. The important effect it captures is how predictions introduce memory to peer dynamics and opinions from time step $t$ impact dynamics at time step $t+1$ through the information distilled into the machine learning model. 

For perfect prediction, the first instance of peer interaction ($t=1$) recovers the classical FJ dynamics anchored at $\xv^{(0)}=\xv^*$. When individuals are not susceptible to platform influence and $K\rightarrow \infty$ it thus holds that~$\xv_\mathrm{PS}=\xv_\mathrm{FJ}$. When $\beta_i>0$ and $t$ grows our model deviates from this classical treatment. The resulting equilibrium will differ qualitatively, as we will show.~\looseness=-1

\subsection{Performative stability}
We characterizing the learning dynamics under perfect prediction. Instantiating our model with the platform policy in Equation~\ref{eq:SL}, we observe that the predictions $\pmb{f}$ admit the following recursion
\begin{align}
    \label{eq:sl_iterative_update}
    \pmb{f}^{(t+1)}&=\Psi_K\left[(I_n-\Db)\pmb{x}^*+\Db\,\pmb{f}^{(t)}\right]\quad \text{with}\quad  \Psi_K\coloneqq\sum_{i=0}^{K-1}(\Da W)^i(I_n-\Da)+(\Da W)^K.
\end{align}
The initial opinions are a linear combination of the current predictions $\pmb{f}^{(t)}$ and the innate opinion $\xv^*$. The matrix 
$\Psi_K$ characterizes peer influence during $K$ consecutive steps of peer interaction, analogous to classical characterizations of opinion dynamics~\citep[e.g.,][]{friedkin1990social}. Recall that $W$ denotes the influence matrix, $K$ the number of peer interaction steps in between model retraining rounds, and $\Lambda_\alpha$ the peer susceptibility matrix.  What is new to our work is an outer loop implied by repeated updates to the platform's  predictions. Interestingly, convergence of our model does not depend on individuals' platform susceptibility, instead, it is sufficient for peer dynamics to be stable.~\looseness=-1 

\begin{proposition}[Existence of a unique performatively stable equilibrium]
\label{prop:existencePS}
Assume the population is characterized by a simple connected graph $\mathcal{G}$, and $\exists i:\alpha_i<1$. Consider a platform that performs perfect prediction of past data. Then, the equilibrium state $\xv_\mathrm{PS}$ exists and is unique.
\end{proposition}
Note that the assumptions coincide with the conditions that guarantee the existence of equilibrium under the FJ model~\citep{parsegov2016novel}. Intuitively, the assumption on the peer susceptibilities implies that there is at least one ``stubborn'' individual that has non-zero weight on their initial belief. The assumption of connectivity makes $\Psi_K$ aperiodic, which implies primitivity and drives the system to convergence. Notably, the result in Proposition~\ref{prop:existencePS} holds for any value of  $\beta_i\in[0,1]$. This implies that our coupled dynamics~\eqref{eq:sl_iterative_update} converge whenever the inner loop dynamics converge. We provide a general closed-form characterization of $\xv_\mathrm{PS}$ in Appendix~\ref{app:proof-perfect-prediction} and discuss particular cases throughout this manuscript.

As we just learned, platform susceptibilities do not matter for the existence of a performatively stable equilibrium, but they critically determine the rate of convergence and the resulting equilibrium opinions. The next result explains how higher platform susceptibility requires more platform retraining steps to reach performative stability.

\begin{proposition}[Convergence under repeated perfect prediction]
\label{prop:convergence_rate_sl}
    Under the same condition as Proposition~\ref{prop:existencePS}, the opinions $\xv_\mathrm{ex}^{(t)}$ under perfect prediction  converge towards $\xv_\mathrm{PS}$ at a linear rate 
    \[\|\pmb{x}_\mathrm{ex}^{(T)}-\pmb{x}_{\mathrm{PS}}\|_2=\mathcal{O}(c^T)\] where $c$ can be characterized as follows:
    \begin{itemize}[topsep=0pt, itemsep=0pt]
        \item If $\beta_i<1\;\forall i\in[n]$ we have $c=\mathcal{O}(\max_i\beta_i)$.
        \item If $\exists i: \beta_i=1$ and $\rho(\Psi_K\Lambda_\beta)<1$ we have $c=\rho(\Psi_K \Db)$.
        \item If $\exists i: \beta_i=1$ and $\rho(\Psi_K\Lambda_\beta)=1$ we have $c=\max\{|\lambda|:\lambda\in \sigma(\Psi_K\Lambda_\beta),|\lambda|<1\}$.
    \end{itemize}
\end{proposition}

Thus, the convergence of the outer loop is governed by the individual who is most susceptible to platform influence. The case distinction arises because for $\beta_i=1$ dynamics are no longer anchored in the innate opinion of individual $i$. Furthermore, since $\Psi_K$ is primitive, the last case needs to be treated differently.  But in general, it holds that higher platform susceptibility slows down convergence which is in line with classical results in performative prediction. These results tells us that high sensitivity $\varepsilon$ of the data distribution to changes in the predictive model implies slower convergence~\citep{PZMH20}. Under our model it can be shown that the  distribution shift is $\varepsilon \propto\|\Lambda_\beta\|$ sensitive in the predictions, see Appendix~\ref{app:distribution_sensitivity}.

\subsection{Homogenizing force of performativity}
Next, we look at how platform influence impacts opinions at equilibrium. We start from the fact that absent performativity, the Friedkin-Johnsen model is known to induce a heterogeneous equilibrium $\xv_\mathrm{FJ}$, whenever innate opinions $\xv^*$ are non-uniform. The following result explains how performativity acts as a homogenizing force.
In the following result the variance of $\mathrm{Var}(\xv_\mathrm{PS})$ serves as a quantifier for opinion heterogeneity and the mean $\mathrm{Mean}(\xv_\mathrm{PS})$ describes aggregate trends in the population. 
\begin{proposition}
\label{prop:mean_variance_sl}
  Consider a network characterized by a simple connected regular graph $\mathcal{G}$, its nodes with peer susceptibilities $\alpha_i\equiv\alpha\in(0,1)$  and platform susceptibilities $\beta_i\equiv\beta\in(0,1)$. The platform  performs perfect prediction of past data. Then, the mean opinion at equilibrium is  given by $\mathrm{Mean}({\xv}_\mathrm{PS})=\frac{1} n \pmb{1}^{\top}\pmb{x}^*$ and the variance satisfies 
  \[\frac{\partial \mathrm{Var}(\pmb{x}_{\mathrm{PS}})}{\partial\beta}< 0.\]
\end{proposition}
The result is a direct consequence of the equilibrium characterization 
$\xv_\mathrm{PS}=(1-\beta)(I_n-\beta\Psi_K )^{-1}\Psi_K\pmb{x}^*$ when $\alpha\in(0,1)$, $\beta<1$, together with the observation that $\Psi_K$ is row-stochastic.

Thus, as platform susceptibility $\beta$ grows the variance of equilibrium opinions decreases.   Remember that for $\beta=0$ the stable point corresponds to $\xv_\mathrm{PS}=\mathrm{FJ}_K(\xv^*)$ and hence the variance reduction can be as large as $\mathrm{Var}(\xv^*)$ for networks with  small  $\alpha$. 
The mean, in turn, remains unaffected by platform interactions. We defer the discussion on the dependence of $\mathrm{Var}(\xv_\mathrm{PS})$ over $\alpha$ to Appendix~\ref{app:var_peer_platform_sus}. 

Before we further characterize the resulting opinion distribution, let us comment on the case $\alpha=1$ excluded by our result. Here  we know that for $K\rightarrow \infty$ the variance of opinions is zero, after the first instance of peer interaction~\citep{degroot1974reaching}. Interestingly, this is not only maintained throughout retraining steps, but we have $\mathrm{Var}(\xv_\mathrm{PS})=0$ under platform influence, even for finite $K$. We refer to the formal result and a separate discussion of the DeGroot model in Appendix~\ref{app:degroot}.

\subsection{Consensus in the limit}
As the heterogeneity of equilibrium opinions gets reduced when $\beta$ grows, a particularly interesting case is $\beta\rightarrow1$. In this case platform influence is maximized.  We show how, in combination with perfect prediction, this induces consensus at equilibrium.

\begin{theorem}
\label{theorem:consensus}
Consider a network characterized by a simple connected graph $\mathcal{G}$, individuals with peer susceptibilities $\alpha_i\in(0,1)\;\forall i$, and a platform that performs perfect prediction of past data. Then, in the limit $\beta_i=\beta\rightarrow 1$ for all $i$,\footnote{we can relax the homogeneity assumption on $\beta$ with an additional technical assumption on the speed at which the limit is attained for individual $\beta_i$, see Appendix~\ref{app:consensus_sl}.} a consensus is reached at equilibrium. Namely, for any $\xv^*$ there exists a constant $c^*=c^*(x^*)\in[0,1]$ such that: 
\begin{equation}
\lim_{\beta\rightarrow 1}(\xv_\mathrm{PS})_i = c^* \quad \forall i\in[n].
\end{equation}

\end{theorem}

Let us first contrast this result with the classical equilibrium under the Friedkin-Johnsen dynamics. The conventional understanding is that as long as $\xv^*$ is heterogeneous and $\alpha_i<1$ for some $i$, the platform-free equilibrium leads to disagreement~\citep{friedkin1990social}. However, with performativity this is different. Predictions repeatedly carry information from the past round into future interactions and, as a result, opinions move closer together over time. When $\beta_i=1$ for all $i$ this results in a consensus at equilibrium \emph{even if} the dynamics are initialized at a heterogeneous $\xv^*$. The reason is that, under strong performativity, the equilibrium at each step is carried over to initialize the next step, which progressively reduces the influence of the innate opinions $\xv^*$ from the first instance. At a high level, this process can be viewed as a gradual transition from FJ dynamics to DeGroot dynamics, as the effect of innate opinions vanishes.

Now, let us contrast the result with the equilibrium in the DeGroot model~\citep{degroot1974reaching}. 
DeGroot can explain consensus at equilibrium for heterogeneous $\xv^*$ by setting $\alpha_i=1$ for all $i$ and making the assumption that $\mathcal G$ is primitive.  
Theorem~\ref{theorem:consensus} can explain convergence to consensus for more general graphs without primitivity. The reason is that anchoring of initial opinions, with $\alpha_i\in(0,1)$, makes the random walk induced by $\Psi_K$ aperiodic.

\section{Indirect peer influence under partial observation} 

Standard opinion dynamic models assume that peer influence travels along the edges of a graph and the network structure governs how opinions flow. Our next result shows that platforms  can add an overlaying structure to the network. In doing so, it can create new pathways for peer influence, linking people who are not otherwise connected in the graph. 

We demonstrate how these alternate pathways for peer influence arise under a simple model of  prediction with limited information. Formally, suppose the platform observes opinions for individuals in a connected subgraph of $\mathcal G$, denoted as $\mathcal{G}'=(\mathcal{V}', \mathcal{E}')$. To make prediction for unobserved individuals, the platform then extrapolated these observations, as detailed in the predictive policy below: 
\begin{align}
    (\pmb{f}^{(t+1)})_i=\begin{cases}
        (\xv_\mathrm{ex}^{(t)})_i,&\text{for}\quad i\in\mathcal{V}',\\
        \frac{1}{|\mathcal{V}'|}\sum_{j\in\mathcal{V}'}(\xv_\mathrm{ex}^{(t)})_j,&\text{otherwise}
    \end{cases}
    \label{eq:mean_estimator}
\end{align}
The algorithm deviates from perfect prediction for unobserved individuals $i\notin\mathcal V'$. To determine their predictions, it gives equal weight to each observed data point. This is a natural choice if no additional information about individuals is available. In cases with additional information, these features could be used to adjust the weighting based on similarity among individuals. In the following result, we use the uniform weight in the proof for simplicity while it can be easily extended to heterogeneous weights case. It shows how such deviations from perfect prediction couple individuals,  even if they are not susceptible to peer influence. 

\begin{proposition}
\label{prop:consensus_mean_estimation}
   Consider a network characterized by a simple connected graph $\mathcal{G}$. Assume the platform has access to expressed opinions of individuals in a  subgraph $\mathcal{G}'=(\mathcal{V}',\mathcal{E}')$ of $\mathcal{G}$ and $\alpha_i\in(0,1)$ for all $ i\in\mathcal{V}'$. Consider innate opinions $\xv^*$ and $\tilde \xv^*$ that differ in a single coordinate $i_0\in\mathcal V'$, and let $\xv_\mathrm{PS}$ and $\tilde \xv_\mathrm{PS}$ be the corresponding stable points when platform predictions follow Equation~\ref{eq:mean_estimator}.  Fix $q\notin\mathcal V'$. Then, even if $\alpha_q=0$, as long as $\beta_q>0$,  it holds that 
\begin{align*}
    |(\xv_\mathrm{PS})_q-(\tilde \xv_\mathrm{PS})_q|>c\beta_q|x_{i_0}^*-\tilde{x}_{i_0}^*|
\end{align*}
where constant $c>0$ is determined by $\Lambda_{\pmb{\alpha}}$, $\mathcal{G}$, and $\mathcal{G}'$.
\end{proposition}

The result explains that changes to the innate opinions of individual $i$ impact the opinion of individual $q$ even if $\alpha_q=0$. Thus, there can be indirect spillover effects through predictive models even if individuals are not susceptible to peer influence.  This is an important  deviation from classical network theory. In the meanwhile, $\mathcal{G}'$ is not necessarily connected, while the platform adds another layer of connectivity among individuals in $\mathcal{V}'$ via the platform-susceptible individuals in $\mathcal{V}/\mathcal{V}'$.

A consequence of these indirect spillover effects is that a consensus can be achieved at equilibrium even in the presence of stubborn individuals.

\begin{theorem}
\label{theorem:consensus_partial_observation}
Consider a network characterized by a simple connected graph $\mathcal{G}$. Assume the platform has access to expressed opinions of individuals in a subgraph $\mathcal{G}'=(\mathcal{V}',\mathcal{E}')$ of $\mathcal{G}$ and $\alpha_i\in(0,1)$ for all $ i\in\mathcal{V}'$. Then, if $\beta_i=1$ for all $i$, the equilibrium $\xv_\mathrm{PS}$ correspond to a consensus. Namely, for any $\xv^*$ there exists a constant $c^*=c^*(x^*)\in[0,1]$ such that: 
\begin{equation}
\lim_{\beta\rightarrow 1}(\xv_\mathrm{PS})_i = c^* \quad \forall i\in[n].
\end{equation} 
\end{theorem}

The proof of Theorem~\ref{theorem:consensus_partial_observation} goes through by observing the platform acts as a virtual node that connects nodes in $\mathcal{G}'$ via the mean estimation. The difference to the result in Theorem~\ref{theorem:consensus} is that we do not require $\alpha_i>0$ for every $i$. Instead,  individuals for whom the platform derives predictions from alternate observations, can have peer susceptibility $\alpha_i=0$. Because rather than through direct peer influence,  they get pulled towards the consensus opinion through predictions instead.

\section{Equilibrium response to platform intervention}

The qualitative differences in how opinions spread across the population through platform and peer interactions affects the equilibrium response to platform interventions. In the following we focus on platform interventions that systematically deviate from perfect prediction for a single individual. As our next result will show, such interventions impact equilibrium opinions across the population in a way that is underestimated when considering peer dynamics in isolation. 

Let us start with classical intuition from social networks. It tells us how platform interventions propagate through the network along social connections. In that way, changes to the initial opinion of individual~$j$, naturally influences any individual $l$ within the same connected component of the graph $\mathcal G$ during subsequent steps of peer interaction. In the same way, individual $l$ can be influenced by platform recommendations given to individual $j\neq l$, even if individual $l$ itself is not susceptible to platform influence. What we show next is that the peer-platform co-evolution amplifies these effects.

To make this concrete, consider a platform that deviates from perfect prediction for a single individual and determines predictions as follows: For a fixed $j\in[n]$ let
\begin{equation}(\fv^{(t+1)})_i=\begin{cases} s& \text{for }i=j \\
(\xv_\mathrm{ex}^{(t)})_i &\text{otherwise}.
\end{cases}
\label{eq:steering}
\end{equation}
In words, the platform intervenes by repeatedly predicting $s\in[0,1]$ for individual $j$ and remains faithful to perfect prediction for all other individuals.

To make this interesting we assume that individual $j$ who is targeted by the platform has platform susceptibility $\beta_j>0$ and we study the indirect platform influence on individual~$l\neq j$. To attribute the observed effect to spillovers, we let  $\beta_l=0$. Further, we let $\xv^*=\mathbf 0$ so that absent intervention we have $\xv_\mathrm{PS}=\mathbf 0$ and any deviation from this consensus state can be attributed to the platform intervention~$s$. We set $\beta_i=\gamma$ to control the platform influence on individuals $i\notin \{j,l\}$ exposed to perfect prediction. The next result explains how the effect of the intervention on equilibrium opinions grows with the platform susceptibility of peers. We refer to Appendix~\ref{app:spillover} for a precise characterization of $\xv_\mathrm{PS}$.

\begin{theorem}[Indirect platform influence] 
\label{prop:spillover}
Consider a network characterized by a simple fully-connected graph $\mathcal{G}$, and peer susceptibilities $\alpha_i\equiv\alpha\in(0,1)$. Let  $\xv^*=\mathbf 0$ and consider a platform that always predicts $s>0$ for randomly sampled individual $j$ according to Equation~\eqref{eq:steering}. Consider $K=\infty$. Let platform susceptibilities satisfy $\beta_i>0$ for $i=j$, $\beta_i=0$ for randomly sampled individual~$i=l$, and $\beta_i=\gamma\in(0,1]$ for $ i\notin \{j,l\}$. Then, the influence of intervention~$s$ applied to individual $j$ on individual~$l$ strictly grows with the platform susceptibility $\gamma$ of their peers. In particular, it holds that 
\[
\frac{\partial\Delta^\gamma_l}{\partial\gamma}>0\quad\text{for}\quad \Delta_l^\gamma:=(\xv_\mathrm{PS})_l-(\xv_\mathrm{PS})_l^{\gamma=0}.\]
Similarly, for the broader population we have 
\begin{align*}
        \frac{\partial \mathrm{Mean}(\xv_\mathrm{PS}) }{\partial \gamma}> 0.
\end{align*}
\end{theorem}

To appreciate the result, note that  $(\xv_\mathrm{PS})_l^{\gamma=0}$ equals the opinion for individual $l$ under the FJ dynamics initialized with $\xv_\mathrm{init}=s\beta_j \ev_j $. This corresponds to the classical network equilibrium under platform intervention. Given the network topology each individual will equally be biased towards $s$ during peer interactions and we have $(\xv_\mathrm{PS})_l^{\gamma=0}=s\beta_j\ev_j^\top \Psi_K \ev_l  >0$.  Now, what Theorem~\ref{prop:spillover} tells us is that the presence of a perfect predictor except for individual $j$ amplifies the indirect influence of the intervention $s$ on the equilibrium opinions. The influence strictly grows with the platform susceptibility $\gamma$ of  peers in the population. As a consequence, classical models of peer interaction that treat platforms as static may miss the amplification of the intervention effect because of this repeated coupling of platform and peer influence.

\section{Experimental results} 
\label{sec:simulations}

To complement our theoretical study, we conduct semi-synthetic experiments on real social network data using popular machine learning methods to derive predictions. For this purpose, we extract features for training the predictive model from public data about individual's profile information. We then simulate our model and investigate how practical deviations from perfect prediction impact dynamics and equilibria.

\vspace{-0.3cm}
\paragraph{Simulation environment.}
We build on the Pokec and Yelp dataset.\footnote{The dataset can be downloaded from~\url{https://snap.stanford.edu/data/soc-Pokec.html} and~\url{https://business.yelp.com/data/resources/open-dataset/}} The Pokec dataset contains anonymized profile data of more than 1.6 million individuals in different modalities together with an underlying social network graph. For our simulations, we sample a connected component of 2,163 nodes which we treat as an undirected graph. The Yelp dataset contains numerous reviews with user profiles describing more than 150,000 businesses. In our simulations, we build an undirected connected network of 1,789 users who are friends with each other and left reviews for the randomly chosen business ``Acme Oyster House''.

\begin{figure*}[t]
\centering
        \begin{subfigure}[b]{0.44\linewidth}
            \centering
    \includegraphics[width=\linewidth]{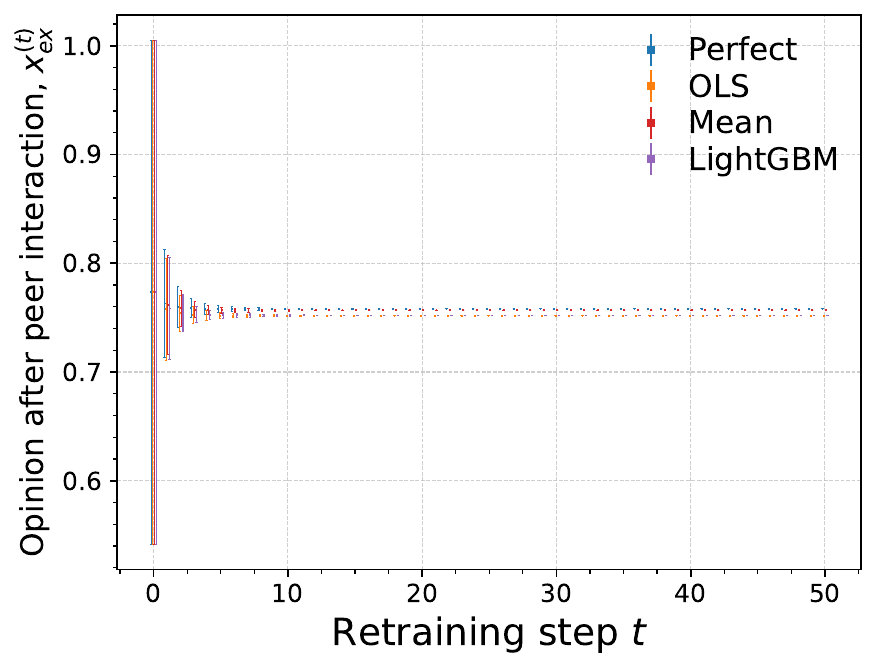}
            \caption{Opinions across retraining steps (Yelp)}
            \label{fig:sl_retrain_steps_yelp}
        \end{subfigure}
        \hfill
            \begin{subfigure}[b]{0.46\linewidth}
            \centering
            \includegraphics[width=\linewidth]{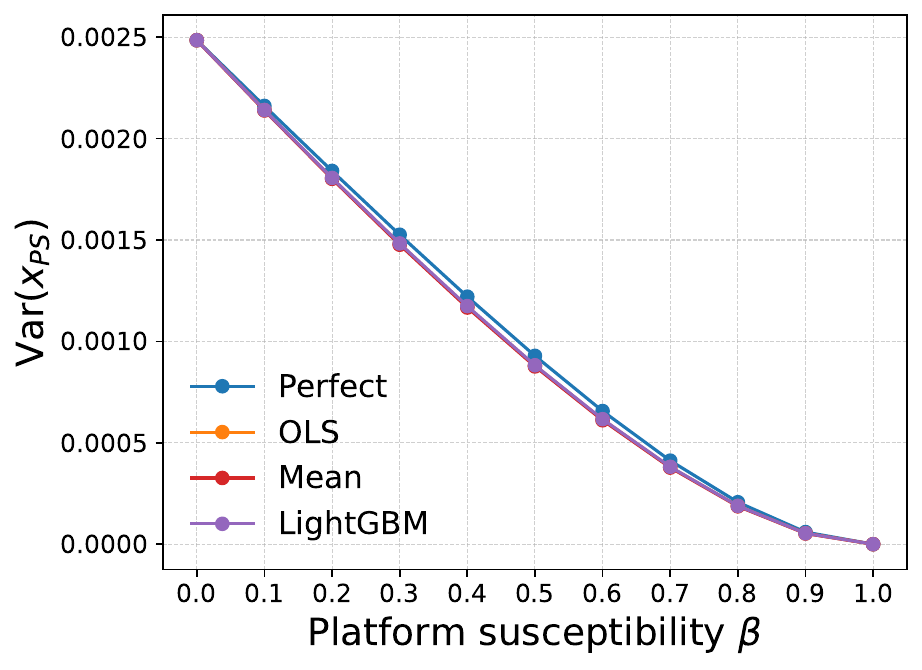}
            \caption{Opinions at equilibrium (Yelp)}
            \label{fig:platform_variance_sus_yelp}
        \end{subfigure}
    \caption{ In $(a)$, we show how performativity homogenizes opinions across retraining steps. The $x$-axis denotes the retraining step $t$, and the $y$-axis corresponds to individuals' expressed opinions in each times step, where the error bars indicate the variance. In (b), we show how platform susceptibility decreases variance of opinions at equilibrium. The $x$-axis denotes the varying homogeneous platform susceptibility while the $y$-axis denotes the variance of opinions in equilibrium. }
    \label{fig:consensus_sl_yelp} 
\end{figure*}

To mimic a machine learning setting we extract node features from the available profile information using a transformer model. For the Pokec network, we choose ``relation$\_$to$\_$smoking'' as the target of prediction. This field contains free form text from which we extract a sentiment scores as a proxy for individuals' opinion towards smoking. For the Yelp dataset, we choose the feature ``stars'' of the reviews as the opinions of users towards the business and normalize it between 0 and 1. We simulate retraining and the dynamics over the network graph (Algorithm~\ref{alg:retrainMain}), treating the opinion extracted from the data as the individuals' innate opinions. For the platform's learning algorithm we contrast perfect prediction (Perfect) with a) mean estimation (Mean), b) Ordinary Least Square (OLS) with predictions clipped to $[0,1]$,  and c) a gradient-boosting method, Light Gradient-Boosting Machine~\citep{ke2017lightgbm}~(LightGBM).

If not stated otherwise, we use the following settings for simulations. We use $K=100$ to approximate convergence of the inner loop dynamics in Algorithm~\ref{alg:retrainMain}. We perform $T=50$ steps of retraining after which the dynamics approximately reached performative stability.  We use a heterogeneous setup where we sample $\{\beta_i\}_{i\in[n]}$ and $\{\alpha_i\}_{i\in[n]}$ from a normal distribution with mean 0.9 and standard deviation 0.1, and the resulting values are clipped within $[0.01, 0.99]$. If the predictions are not perfect, we assume opinions of 80\% of the nodes are observed and use it as the training dataset to predict the opinions of the remaining 20\% population, and we denote these two sets of individuals as $O$ and $U$ respectively. All the simulations are conducted on a commodity hardware using Python 3.13. The codes for the plots can be accessed via \href{https://github.com/wujiduan/reaching-a-consensus-in-predictive-loops.git}{our repository}. The following figures in this section is generated using Yelp dataset. Simulations with Pokec dataset and further results on peer susceptibility are provided in Appendix~\ref{app:more_simulations}.

\begin{figure}
        \begin{subfigure}[b]{0.45\linewidth}
            \centering
            \includegraphics[width=\linewidth]{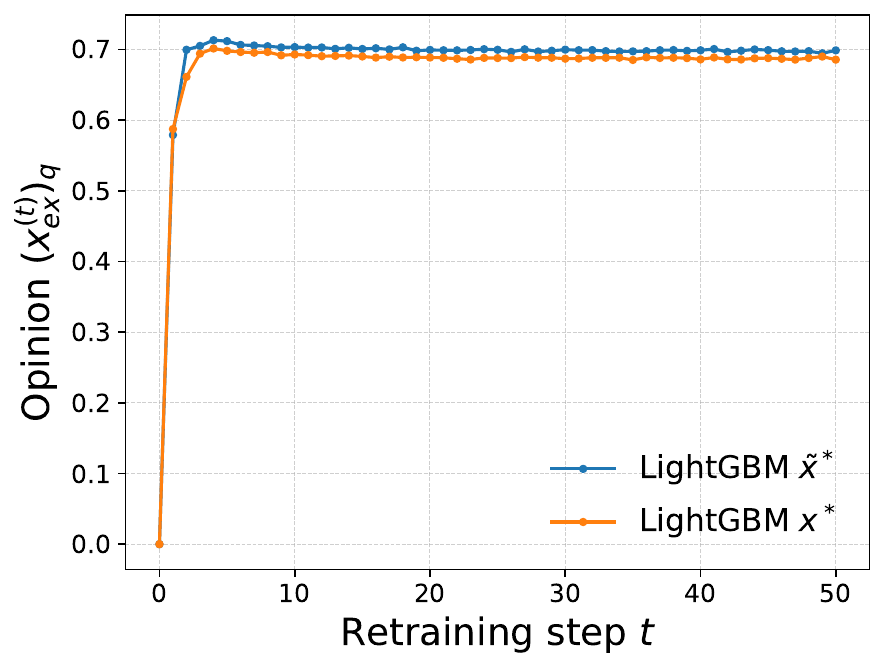}
            \caption{Evolving opinion of stubborn individual (Yelp)}
            \label{fig:stubborn_yelp}
        \end{subfigure}
        \hfill
        \begin{subfigure}[b]{0.45\linewidth}
            \centering
            \includegraphics[width=\linewidth]{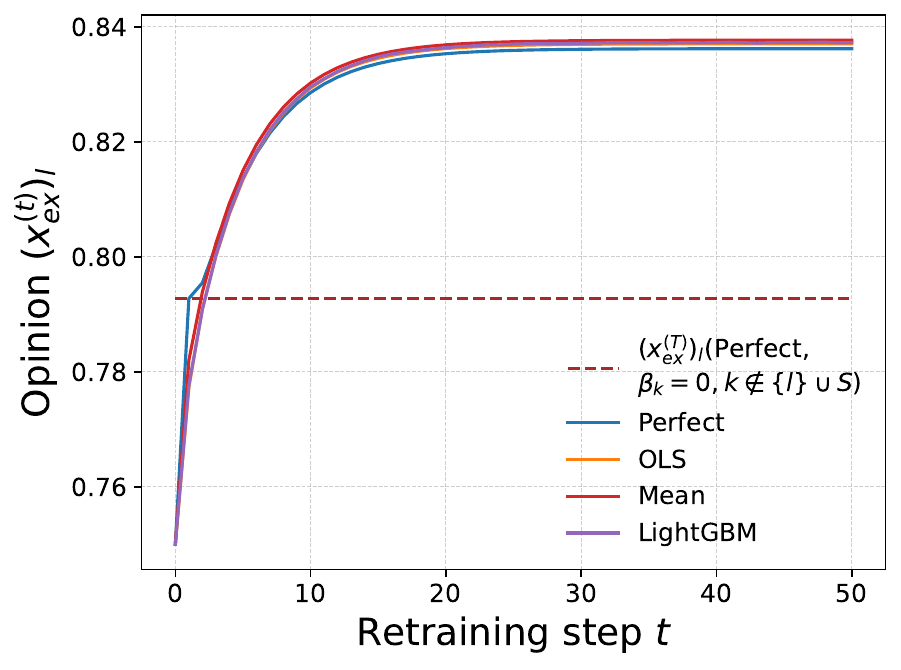}
            \caption{Indirect effect of platform intervention (Yelp) }
            \label{fig:steer_yelp}
        \end{subfigure}
        \caption{ In $(a)$, we show how a stubborn individual $q\in U$ with $\alpha_q=0$ and $x^*_q=0$ is influenced by their peers across retraining steps for two different scenarios. In $(b)$, we show how indirect platform influence increases over retraining steps. We compare different the opinion of the stubborn individual $l$ with $\beta_l=0$. The purple dashed line denotes the opinion of the stubborn individual if we set~$\beta_k=0$,~$k\notin \{l\}\cup S$. }
\end{figure}

\vspace{-0.3cm}

\subsection{Retraining and equilibria}

We visualize how opinions evolve across retraining steps for the different learning settings. In Figure~\ref{fig:sl_retrain_steps_yelp}, we plot the mean and variance of opinions $\xv_\mathrm{ex}^{(t)}$ across retraining steps for $\beta_i=1$ and heterogeneous peer susceptibilities.  We see that for all four predictive strategies, opinions converge, and the variance approximately converges to $0$ as $t$ increases. This illustrates convergence beyond perfect prediction, and suggests that our intuition of performativity as an explanation of consensus is not specific to perfect prediction, but also extends to parametric predictors.

In Figure~\ref{fig:platform_variance_sus_yelp} we again hold the heterogeneous peer susceptibilities fixed, but replace platform susceptibilities by $\beta_i\equiv\beta$, and vary $\beta$ between $0$ and $1$. We can see that the variance of the equilibrium opinion decreases as $\beta$ grows. This is evaluated based on the closed form expression for the equilibrium with perfect prediction. But note that the Yelp network is not regular which suggests that the take-away for Proposition~\ref{prop:mean_variance_sl} is not bound to a specific network structure.

\vspace{-0.3cm}

\subsection{Indirect peer influence}

Next, we show how platform interactions creates indirect peer influence beyond mean estimation. In particular, we extrapolate predictions to unseen individuals using parametric models. In line with Theorem~\ref{prop:consensus_mean_estimation}, we randomly sample an individual~$q$ in unobserved group $U$ in the population and set $\alpha_q=0$, $x_q^*=0$. We then randomly select $10\%$ of individuals in observed group $O$ and set their innate opinions to $1$, and we denote this new innate opinions as $\tilde{\xv}^*$. With the new innate opinions $\tilde{\xv}^*$, we observe how the opinion $(\xv_\mathrm{ex}^{(t)})_q$ converge to a value different from the equilibrium opinion under the original innate opinions $\xv^*$.
We use the LightGBM predictions to visualize this in Figure~\ref{fig:stubborn_yelp}. Empirical results of mean estimation and OLS are deferred to Appendix~\ref{app:more_simulations}.

Next, we show how peer interactions amplify platform influence, as analyzed in Theorem~\ref{prop:spillover}. To this end, we use simulations on the Pokec network, which is not fully connected. We also adopted the same $\{\alpha_i\}$ and $\xv^*$ as in producing Figure~\ref{fig:sl_retrain_steps_yelp}. We randomly choose an individual~$l$ and set $\beta_l=0$. Then, instead of one individual, we randomly sample a set $S$ that contains $10\%$ of individuals of the network and steer them towards $s=1$. In Figure~\ref{fig:steer_yelp}, we observe the opinion of individual~$l$ increasingly deviates from $x_l^*$ and $(\xv_\mathrm{ex}^{(1)})_l$ as retraining step $t$ increases. The purple dashed line shows the setting where $\beta_i=0$ for individuals who are not steered, i.e., $i\notin S$, this corresponds to the outcomes after a single instance of peer interactions. We see how opinions deviate from this state through repeated peer-platform interactions.

\section{Discussion}
\label{sec:future_direction}

With theoretical results and simulations we demonstrated how predictive loops can qualitatively alter equilibria in social networks. We identify two key mechanisms that lead to divergent behavior. First, learning systems introduce memory into peer dynamics by deriving future predictions from past observations, thereby coupling opinions across time. Second, learning systems create new pathways for peer influence beyond the topology of social networks by repeatedly aggregating data across a population and distilling it into personalized predictions. Using a minimal model, we provide intuition for how these mechanisms interact with peer dynamics and alter equilibrium outcomes. While the number of unexplored questions at the intersection of performative prediction and network science, ranging from empirical studies to theoretical understanding, is vast, there are several natural extensions of our model that we discuss here.

For example, in our model we abstract learning systems as non-parametric predictors that have access to perfect measurements of opinions. Extending the study to incorporate measurement uncertainty, stochasticity, or limited information, as well as parametric prediction functions would illuminate other practical dimensions of performative effects. Our simulations suggest that many of the qualitative findings hold beyond perfect prediction. Further, we assume that performative effects are self-fulfilling, and leave the network topology, as well as peer susceptibilities untouched. However, related work suggests that platforms can also change characteristics of the network~\citep[][e.g.,]{zhou2024modeling, chitra2019understanding}. It would be interesting to investigate how such alternate assumptions of platform influence interplay with predictive loops. In that context, interventions at the level of susceptibilities~\citep{abebe18persuade} would also be interesting to study. Similarly, our results suggest that peer dynamics can amplify platform effects. Rethinking established influence maximization results while keeping this co-evolution in mind might help develop alternative explanations for reconciling empirical findings with theoretical models and explain the power digital platforms hold over individuals.

More broadly, the need to understand how digital platforms reshape our society is not only scientific but also regulatory. Recognizing the important role that emerging dynamics involving platforms and users play, the Digital Services Act explicitly calls for greater transparency~\citep{oecd25}:
\begin{quote}
``[...] to shed light on how the algorithms and processes deployed by these platforms influence the way information flows in our society, and influence individual platform users'' 
\end{quote}
The effect of platforms over users is also a central concern in digital antitrust~\citep{stigler19}. Linking our insights to measures of performative power~\citep{hardt2022power,mendler2024engine} could help support future investigations into the economic power of digital platforms. In particular, our results suggest that standard causal designs might underestimate the performative power of digital platforms. Thus, gathering valid empirical estimates of platform effects, keeping network effect in mind, would be an important next step.

\section*{Acknowledgements}
Rediet Abebe and Celestine Mendler Dünner acknowledge the financial support of the Hector Foundation. Abebe was supported by the Andrew Carnegie Fellowship Program. Jiduan Wu would like to thank for the financial support from the Max Planck ETH Center for Learning Systems (CLS). We would like to thank Yatong Chen, Ana-Andreea Stoica, and anonymous reviewers for helpful feedback on the manuscript.

\bibliographystyle{ACM-Reference-Format}
\bibliography{ref}

@article{friedkin1990social,
  title={Social influence and opinions},
  author={Friedkin, Noah E and Johnsen, Eugene C},
  journal={Journal of mathematical sociology},
  volume={15},
  number={3-4},
  pages={193--206},
  year={1990},
  publisher={Taylor \& Francis}
}

@article{degroot1974reaching,
  title={Reaching a consensus},
  author={DeGroot, Morris H},
  journal={Journal of the American Statistical association},
  volume={69},
  number={345},
  pages={118--121},
  year={1974},
  publisher={Taylor \& Francis}
}

@incollection{friedkin1999social,
  author    = {Noah E. Friedkin and Eugene C. Johnsen},
  title     = {Social Influence Networks and Opinion Change},
  booktitle = {Advances in Group Processes},
  year      = {1999},
  volume    = {16},
  pages     = {1--29},
  publisher = {Emerald},
}

@article{rainer2002opinion,
	author = {Hegselmann Rainer and Ulrich Krause},
	journal = {Journal of Artificial Societies and Social Simulation},
	number = {3},
	title = {Opinion Dynamics and Bounded Confidence: Models, Analysis and Simulation},
	volume = {5},
	year = {2002}
}

@article{perra2019modelling,
  title={Modelling opinion dynamics in the age of algorithmic personalisation},
  author={Perra, Nicola and Rocha, Luis EC},
  journal={Scientific reports},
  volume={9},
  number={1},
  pages={7261},
  year={2019},
  publisher={Nature Publishing Group UK London}
}

@article{huszar2022algorithmic,
  title={Algorithmic amplification of politics on Twitter},
  author={Husz{\'a}r, Ferenc and Ktena, Sofia Ira and O’Brien, Conor and Belli, Luca and Schlaikjer, Andrew and Hardt, Moritz},
  journal={Proceedings of the national academy of sciences},
  volume={119},
  number={1},
  pages={e2025334119},
  year={2022},
  publisher={National Academy of Sciences}
}

@incollection{peralta2025opinion,
  title={Opinion dynamics in social networks: From models to data},
  author={Peralta, Antonio F and Kert{\'e}sz, J{\'a}nos and I{\~n}iguez, Gerardo},
  booktitle={Handbook of Computational Social Science},
  pages={384--406},
  year={2025},
  publisher={Edward Elgar Publishing Limited}
}

@InProceedings{cheng2023causal,
  title = 	 {Causal Inference out of Control: Estimating Performativity without Treatment Randomization},
  author =       {Cheng, Gary and Hardt, Moritz and Mendler-D\"{u}nner, Celestine},
  booktitle = 	 {International Conference on Machine Learning},
  pages = 	 {8077--8103},
  year = 	 {2024}
}

@misc{stigler19,
	author = {{Stigler Committee}},
		month = {September},
	title = {Final Report: Stigler Committee on Digital Platforms},
	year = {2019}}

@inproceedings{PZMH20,
        author = {Perdomo, Juan C. and Zrnic, Tijana and Mendler-D\"{u}nner, Celestine and Hardt, Moritz},
        booktitle={International Conference on Machine Learning},
        pages = {7599-7609},
        title = {Performative Prediction},
        volume = {119},
        year = {2020},
        organization={PMLR}
}

@inproceedings{mendler20stochasticPP,
        author = {Mendler-D\"{u}nner, Celestine and Perdomo, Juan and Zrnic, Tijana and Hardt, Moritz},
        booktitle = {Advances in Neural Information Processing Systems},
        pages = {4929--4939},
        title = {Stochastic Optimization for Performative Prediction},
        volume = {33},
        year = {2020}
}

@inproceedings{miller2021outside,
        title={Outside the echo chamber: Optimizing the performative risk},
        author={Miller, John P and Perdomo, Juan C and Zrnic, Tijana},
        booktitle={International Conference on Machine Learning},
        pages={7710--7720},
        year={2021},
        organization={PMLR}
}

@inproceedings{jag21alt,
        author = {Meena Jagadeesan and Celestine Mendler{-}D{\"{u}}nner and Moritz Hardt},
        booktitle = {International Conference on Machine Learning},
        title = {Alternative Microfoundations for Strategic Classification},
        year = {2021},
        pages = 	 {4687--4697},
        volume = 	 {139},
        publisher = {PMLR}
}

@inproceedings{mendler22causal,
        title={Anticipating Performativity by Predicting from Predictions},
        author={Celestine Mendler-D{\"u}nner and Frances Ding and Yixin Wang},
        booktitle={Advances in Neural Information Processing Systems},
        year={2022},
        articleno = {2260}
        }

@article{drusvyatskiy23stochastic,
author = {Drusvyatskiy, Dmitriy and Xiao, Lin},
title = {Stochastic Optimization with Decision-Dependent Distributions},
journal = {Mathematics of Operations Research},
volume = {48},
number = {2},
pages = {954-998},
year = {2023},
}

@InProceedings{izzo21gd,
  title = 	 {How to Learn when Data Reacts to Your Model: Performative Gradient Descent},
  author =       {Izzo, Zachary and Ying, Lexing and Zou, James},
  booktitle = 	 {International Conference on Machine Learning},
  pages = 	 {4641--4650},
  year = 	 {2021},
  volume = 	 {139},
  publisher =    {PMLR}
}

@inproceedings{hardt2022power,
  title={Performative power},
  author={Hardt, Moritz and Jagadeesan, Meena and Mendler-D{\"u}nner, Celestine},
  booktitle={Advances in Neural Information Processing Systems},
  volume={35},
  pages={22969--22981},
  year={2022}
}

@inproceedings{hardt16strategic,
author = {Hardt, Moritz and Megiddo, Nimrod and Papadimitriou, Christos and Wootters, Mary},
title = {Strategic Classification},
year = {2016},
booktitle = {ACM Conference on Innovations in Theoretical Computer Science},
pages = {111–122},
numpages = {12},
}

@article{wang2023constrained,
  title={Constrained optimization with decision-dependent distributions},
  author={Wang, Zifan and Liu, Changxin and Parisini, Thomas and Zavlanos, Michael M and Johansson, Karl H},
  journal={IEEE Transactions on Automatic Control},
  year={2025},
  publisher={IEEE}
}

@InProceedings{mendler2024engine,
  title={An engine not a camera: Measuring performative power of online search},
  author={Mendler-D{\"u}nner, Celestine and Carovano, Gabriele and Hardt, Moritz},
  booktitle={Advances in Neural Information Processing Systems},
  year={2024}
}

@article{hardt25sts,
author = {Moritz Hardt and Celestine Mendler-Dünner},
title = {Performative Prediction: Past and Future},
journal = {Statistical Science},
volume = {40},
number = {3},
pages = {417-436},
year = {2025}
}

@book{pariser2011filter,
  title={The filter bubble: How the new personalized web is changing what we read and how we think},
  author={Pariser, Eli},
  year={2011},
  publisher={Penguin}
}

@article{adomavicius2013recommender,
  title={Do recommender systems manipulate consumer preferences? A study of anchoring effects},
  author={Adomavicius, Gediminas and Bockstedt, Jesse C and Curley, Shawn P and Zhang, Jingjing},
  journal={Information Systems Research},
  volume={24},
  number={4},
  pages={956--975},
  year={2013},
  publisher={INFORMS}
}

@article{cinelli2021echo,
  title={The echo chamber effect on social media},
  author={Cinelli, Matteo and De Francisci Morales, Gianmarco and Galeazzi, Alessandro and Quattrociocchi, Walter and Starnini, Michele},
  journal={Proceedings of the national academy of sciences},
  volume={118},
  number={9},
  pages={e2023301118},
  year={2021},
  publisher={National Academy of Sciences}
}

@article{del2016spreading,
  title={The spreading of misinformation online},
  author={Del Vicario, Michela and Bessi, Alessandro and Zollo, Fabiana and Petroni, Fabio and Scala, Antonio and Caldarelli, Guido and Stanley, H Eugene and Quattrociocchi, Walter},
  journal={Proceedings of the national academy of Sciences},
  volume={113},
  number={3},
  pages={554--559},
  year={2016},
  publisher={National Academy of Sciences}
}

@article{holley1975ergodic,
  title={Ergodic theorems for weakly interacting infinite systems and the voter model},
  author={Holley, Richard A and Liggett, Thomas M},
  journal={The annals of probability},
  pages={643--663},
  year={1975},
  publisher={JSTOR}
}

@article{yildiz2013binary,
  title={Binary opinion dynamics with stubborn agents},
  author={Yildiz, Ercan and Ozdaglar, Asuman and Acemoglu, Daron and Saberi, Amin and Scaglione, Anna},
  journal={ACM Transactions on Economics and Computation (TEAC)},
  volume={1},
  number={4},
  pages={1--30},
  year={2013},
  publisher={ACM New York, NY, USA}
}

@misc{oecd25,
year ={2025},
title ={Social Media Governance Policy Brief: How the DSA can enable a public science of digital platform social impacts},
author = {OECD},
note = {\url{https://oecd.ai/en/wonk/documents/social-media-governance-policy-brief-how-the-dsa-can-enable-a-public-science-of-digital-platform-social-impacts-policy-brief}}
}

@article{jia2015opinion,
  title={Opinion dynamics and the evolution of social power in influence networks},
  author={Jia, Peng and MirTabatabaei, Anahita and Friedkin, Noah E and Bullo, Francesco},
  journal={SIAM review},
  volume={57},
  number={3},
  pages={367--397},
  year={2015},
  publisher={SIAM}
}

@article{friedkin2016theory,
  title={A theory of the evolution of social power: Natural trajectories of interpersonal influence systems along issue sequences},
  author={Friedkin, Noah E and Jia, Peng and Bullo, Francesco},
  journal={Sociological Science},
  volume={3},
  pages={444--472},
  year={2016}
}

@article{proskurnikov2017tutorial,
  title={A tutorial on modeling and analysis of dynamic social networks. Part I},
  author={Proskurnikov, Anton V and Tempo, Roberto},
  journal={Annual Reviews in Control},
  volume={43},
  pages={65--79},
  year={2017},
  publisher={Elsevier}
}

@article{jia2019opinion,
  title={Opinion dynamics and social power evolution: A single--timescale model},
  author={Jia, Peng and Friedkin, Noah E and Bullo, Francesco},
  journal={IEEE Transactions on Control of Network Systems},
  volume={7},
  number={2},
  pages={899--911},
  year={2019},
  publisher={IEEE}
}

@article{wang2023relationship,
  title={On the relationship between relevance and conflict in online social link recommendations},
  author={Wang, Yanbang and Kleinberg, Jon},
  journal={Advances in Neural Information Processing Systems},
  volume={36},
  pages={36708--36725},
  year={2023}
}

@inproceedings{zhou2024modeling,
  title={Modeling the impact of timeline algorithms on opinion dynamics using low-rank updates},
  author={Zhou, Tianyi and Neumann, Stefan and Garimella, Kiran and Gionis, Aristides},
  booktitle={Proceedings of the ACM Web Conference 2024},
  pages={2694--2702},
  year={2024}
}

@article{mei22median,
  title = {Micro-foundation of opinion dynamics: Rich consequences of the weighted-median mechanism},
  author = {Mei, Wenjun and Bullo, Francesco and Chen, Ge and Hendrickx, Julien M. and D\"orfler, Florian},
  journal = {Phys. Rev. Res.},
  volume = {4},
  issue = {2},
  pages = {023213},
  numpages = {10},
  year = {2022}
}

@article{chitra2019understanding,
  title={Understanding filter bubbles and polarization in social networks},
  author={Chitra, Uthsav and Musco, Christopher},
  journal={arXiv preprint arXiv:1906.08772},
  year={2019}
}

@book{seneta2006non,
  title={Non-negative matrices and Markov chains},
  author={Seneta, Eugene},
  year={2006},
  publisher={Springer Science \& Business Media}
}

@article{pescetelli2022bots,
  title={Bots influence opinion dynamics without direct human-bot interaction: the mediating role of recommender systems},
  author={Pescetelli, Niccolo and Barkoczi, Daniel and Cebrian, Manuel},
  journal={Applied Network Science},
  volume={7},
  number={1},
  pages={46},
  year={2022},
  publisher={Springer}
}

@inproceedings{wang2023network,
  title={Network effects in performative prediction games},
  author={Wang, Xiaolu and Yau, Chung-Yiu and Wai, Hoi To},
  booktitle={International Conference on Machine Learning},
  pages={36514--36540},
  year={2023},
  organization={PMLR}
}

@book{berman1994nonnegative,
  title={Nonnegative matrices in the mathematical sciences},
  author={Berman, Abraham and Plemmons, Robert J},
  year={1994},
  publisher={SIAM}
}

@article{
bail18polarize,
author = {Christopher A. Bail  and Lisa P. Argyle  and Taylor W. Brown  and John P. Bumpus  and Haohan Chen  and M. B. Fallin Hunzaker  and Jaemin Lee  and Marcus Mann  and Friedolin Merhout  and Alexander Volfovsky },
title = {Exposure to opposing views on social media can increase political polarization},
journal = {Proceedings of the National Academy of Sciences},
volume = {115},
number = {37},
pages = {9216-9221},
year = {2018}
}

@article{
santos21link,
author = {Fernando P. Santos  and Yphtach Lelkes  and Simon A. Levin },
title = {Link recommendation algorithms and dynamics of polarization in online social networks},
journal = {Proceedings of the National Academy of Sciences},
volume = {118},
number = {50},
pages = {e2102141118},
year = {2021}
}

@article{narayana15click,
author = {Narayanan, Sridhar and Kalyanam, Kirthi},
title = {Position Effects in Search Advertising and their Moderators: A Regression Discontinuity Approach},
year = {2015},
issue_date = {May 2015},
publisher = {INFORMS},
address = {Linthicum, MD, USA},
volume = {34},
number = {3},
journal = {Marketing Science},
month = may,
pages = {388–407},
numpages = {20}
}

@Article{ursu18expedia,
journal={Marketing Science},
author={Raluca M. Ursu},
title={The Power of Rankings: Quantifying the Effect of Rankings on Online Consumer Search and Purchase Decisions},
year={2018},
month={August},
pages={530-552},
volume={37},
number={4}
}

@article{anderson12yelp,
author = {Anderson, Michael and Magruder, Jeremy},
title = {Learning from the Crowd: Regression Discontinuity Estimates of the Effects of an Online Review Database},
journal = {The Economic Journal},
volume = {122},
number = {563},
pages = {957-989},
year = {2012}
}

@InProceedings{mehrnaz23pp,
  title = 	 {Performative Prediction with Neural Networks},
  author =       {Mofakhami, Mehrnaz and Mitliagkas, Ioannis and Gidel, Gauthier},
  booktitle = 	 {Proceedings of The 26th International Conference on Artificial Intelligence and Statistics},
  pages = 	 {11079--11093},
  year = 	 {2023},
  volume = 	 {206},
  series = 	 {Proceedings of Machine Learning Research},
  month = 	 {25--27 Apr},
  publisher =    {PMLR}
}

@inproceedings{
saig2025evolutionary,
title={Evolutionary Prediction Games},
author={Eden Saig and Nir Rosenfeld},
booktitle={The Thirty-ninth Annual Conference on Neural Information Processing Systems},
year={2025}
}

@article{narang23multi,
  author  = {Adhyyan Narang and Evan Faulkner and Dmitriy Drusvyatskiy and Maryam Fazel and Lillian J. Ratliff},
  title   = {Multiplayer Performative Prediction: Learning in Decision-Dependent Games},
  journal = {Journal of Machine Learning Research},
  year    = {2023},
  volume  = {24},
  number  = {202},
  pages   = {1--56}
}

@article{gauthier2026political,
  title={The political effects of X’s feed algorithm},
  author={Gauthier, Germain and Hodler, Roland and Widmer, Philine and Zhuravskaya, Ekaterina},
  journal={Nature},
  pages={1--8},
  year={2026},
  publisher={Nature Publishing Group UK London}
}

@article{parsegov2016novel,
  title={Novel multidimensional models of opinion dynamics in social networks},
  author={Parsegov, Sergey E and Proskurnikov, Anton V and Tempo, Roberto and Friedkin, Noah E},
  journal={IEEE Transactions on Automatic Control},
  volume={62},
  number={5},
  pages={2270--2285},
  year={2016},
  publisher={IEEE}
}

@article{farina2026stability,
  title={The Stability of Online Algorithms in Performative Prediction},
  author={Farina, Gabriele and Perdomo, Juan Carlos},
  journal={arXiv preprint arXiv:2602.24207},
  year={2026}
}

@inproceedings{abebe18persuade,
author = {Abebe, Rediet and Kleinberg, Jon and Parkes, David and Tsourakakis, Charalampos E.},
title = {Opinion Dynamics with Varying Susceptibility to Persuasion},
year = {2018},
isbn = {9781450355520},
publisher = {Association for Computing Machinery},
address = {New York, NY, USA},
booktitle = {Proceedings of the 24th ACM SIGKDD International Conference on Knowledge Discovery \& Data Mining},
pages = {1089–1098},
numpages = {10},
series = {KDD '18}
}

@inproceedings{out2024impact,
  title={The impact of external sources on the friedkin--johnsen model},
  author={Out, Charlotte and Tu, Sijing and Neumann, Stefan and Zehmakan, Ahad N},
  booktitle={Proceedings of the 33rd ACM International Conference on Information and Knowledge Management},
  pages={1815--1824},
  year={2024}
}

@inproceedings{eilat2023performative,
  title={Performative recommendation: diversifying content via strategic incentives},
  author={Eilat, Itay and Rosenfeld, Nir},
  booktitle={International Conference on Machine Learning},
  pages={9082--9103},
  year={2023},
  organization={PMLR}
}

@inproceedings{chaney2018algorithmic,
  title={How algorithmic confounding in recommendation systems increases homogeneity and decreases utility},
  author={Chaney, Allison JB and Stewart, Brandon M and Engelhardt, Barbara E},
  booktitle={Proceedings of the 12th ACM conference on recommender systems},
  pages={224--232},
  year={2018}
}

@article{ke2017lightgbm,
  title={Lightgbm: A highly efficient gradient boosting decision tree},
  author={Ke, Guolin and Meng, Qi and Finley, Thomas and Wang, Taifeng and Chen, Wei and Ma, Weidong and Ye, Qiwei and Liu, Tie-Yan},
  journal={Advances in neural information processing systems},
  volume={30},
  year={2017}
}

@article{curry2014prevalence,
  title={Prevalence of internet and social media usage in orthopedic surgery},
  author={Curry, Emily and Li, Xinning and Nguyen, Joseph and Matzkin, Elizabeth},
  journal={Orthopedic reviews},
  volume={6},
  number={3},
  pages={5483},
  year={2014}
}

@article{farnan2013online,
  title={Online medical professionalism: patient and public relationships: policy statement from the American College of Physicians and the Federation of State Medical Boards},
  author={Farnan, Jeanne M and Snyder Sulmasy, Lois and Worster, Brooke K and Chaudhry, Humayun J and Rhyne, Janelle A and Arora, Vineet M and American College of Physicians Ethics, Professionalism and Human Rights Committee and American College of Physicians Council of Associates and Federation of State Medical Boards Special Committee on Ethics and Professionalism*},
  journal={Annals of internal medicine},
  volume={158},
  number={8},
  pages={620--627},
  year={2013},
  publisher={American College of Physicians}
}

\newpage

\appendix

\section{Proofs}
\label{app:proofs}
Throughout the appendix, we say $A\preceq B$ for two arbitrary symmetric matrices $A,B$, if $B-A$ is positive semi-definite. The notation $\pmb{0}$, $\pmb{1}$ denotes a matrix or vector of appropriate size with all $0$ or $1$ entries respectively, and $I$ denotes the identity matrix.

\subsection{Auxiliary results}

\begin{lemma}
\label{lemma:primitive_matrix}
    (Primitive matrices) Let $M\in\mathbb{R}^{n\times n}$ be irreducible and non-negative. Further let   $D\in\mathbb{R}^{n\times n}$ be a diagonal matrix with non-negative diagonal entries and at least one positive entry. Then, the sum of the two matrices  $M+D$ is primitive. 
\end{lemma}
\begin{proof}
Suppose for $i_0\in[n]$, $D_{i_0,i_0}>0$. The proof is immediate by noticing that the matrix $M+D$ is irreducible and aperiodic by adding a self-loop.
\end{proof}

\begin{lemma}
\label{corollary:perron_frobenius_coro}
Let $M\in\mathbb{R}^{n\times n}$ be an arbitrary primitive and row-stochastic matrix, then the sequence of matrices $\{M^t\}_{t\in\mathbb{N}}$ converges as follows
\begin{align*}
    \lim_{t\rightarrow +\infty}M^t=\mathbf{1}\pmb{y}_M^{\top}.
\end{align*}
where $\pmb{y}_M$ is the eigenvector of $M^{\top}$ with respect to eigenvalue $1$ and $\sum_{i=1}^n y_i=1$.
\end{lemma}
\begin{proof}
The result is a corollary of the Perron-Frobenius theorem and Theorem 1.2 in~\cite{seneta2006non}.
From~\cite{seneta2006non} we have
\begin{align*}
    M^t=\pmb{1}\pmb{y}_M^{\top}+\mathcal{O}(|\lambda_2|^t)
\end{align*}
where $\{\lambda_i\}_{i=1}^e$ are eigenvalues of $M$ such that $\lambda_1=1>|\lambda_2|\geq|\lambda_3|\geq\cdots\geq |\lambda_e|$. Thus, the second term vanishes as $t\rightarrow \infty$.
\end{proof}

\subsection{Proof of Proposition~\ref{prop:existencePS}} 
\label{app:proof-perfect-prediction}

\begin{proof}

    To characterize the equilibrium under repeated perfect prediction, we note that 
    \begin{align}
        \pmb{x}_\mathrm{ex}^{(t+1)}&=\Psi_K\pmb{x}^{(t)}_\mathrm{init}\notag\\&=\Psi_K((I_n-\Lambda_{\beta})\pmb{x}^*+\Lambda_{\beta}\pmb{f}^{(t)})\notag\\
        &=\Psi_K((I_n-\Lambda_{\beta})\pmb{x}^*+\Lambda_{\beta}\pmb{x}_\mathrm{ex}^{(t)}),\quad t\geq 0.
        \label{eq:iterative_update}
    \end{align}
    From this recursion we can see that for proving the existence of a stable point we need to be able to control $\lim_{t\rightarrow\infty} (\Psi_K\Lambda_\beta)^t$. 
    
    Applying Corollary~1 in \citet{parsegov2016novel}, we know $\Psi_\infty$ is well-defined. Further, for any $K\in\{\infty\}\cup\mathbb{N}$ the matrix $\Psi_K$ is row-stochastic by observing 
    \begin{align*}
        \Psi_K\pmb{1}&=\left[\sum_{i=0}^{K-1}(\Lambda_{\alpha}W)^i(I_n-\Lambda_{\alpha})+(\Lambda_{\alpha}W)^K\right] \pmb{1}\\
        &=(\Lambda_{\alpha}W)^{K-1}\Lambda_{\alpha}\mathbf{1}+(\Lambda_{\alpha}W)^{K-1}(I_n-\Lambda_{\alpha})\mathbf{1}\\
        &\quad+\sum_{j=0}^{K-2}(\Lambda_{\alpha}W)^j(I_n-\Lambda_{\alpha})\pmb{1}\\
        &=(\Lambda_{\alpha}W)^{K-2}\Lambda_{\alpha} \pmb{1}+\sum_{j=0}^{K-2}(\Lambda_{\alpha}W)^j(I_n-\Lambda_{\alpha})\pmb{1}\\
        &=\cdots =\pmb{1},\quad K<\infty,
    \end{align*}
    because  $W$ is row-stochastic. For $K=\infty$ we have $\Psi_{\infty}=\lim_{k\rightarrow\infty}\Psi_k$ and hence $\Psi_{\infty}$ is also row-stochastic. Perron-Frobenius theorem implies $\rho(\Psi_K)\leq 1$.
    
    Now, we take a detour and observe, under permutation of the node indices, $\Psi_K$ can be decomposed as follows
    \begin{align*}
        \Psi_K=\begin{bmatrix}
            I_{n'}&\pmb{0}\\
            (\Psi_K)_{21}&(\Psi_K)_{22}
        \end{bmatrix}
    \end{align*}
    where $n'=|\{i:i\in[n],\alpha_i=0\}|$ denotes the number of individuals who are fully stubborn towards their peers. With this decomposition, we'll show the existence of $\xv_\mathrm{PS}$ by proving $\lim_{t\rightarrow\infty}(\Psi_\infty)^t$ exists first.
    \begin{itemize}
        \item When $n'>0$, if $\rho((\Psi_K)_{22})<1$ the proof is immediate, and it remains to discuss $\rho((\Psi_K)_{22})=1$. Here, we apply Lemma~2 in \citet{parsegov2016novel} to $(\Psi_K)_{22}$, and under node permutation, we have,
    \begin{align*}
        (\Psi_{K})_{22}=\begin{bmatrix}
            ((\Psi_K)_{22})_{11}&((\Psi_K)_{22})_{12}\\
            \pmb{0}&((\Psi_K)_{22})_{22}
        \end{bmatrix}
    \end{align*}
    where block $((\Psi_K)_{22})_{11}$ contains indices $P:=\{i:i\in[n]/[n'],\exists_j\in[n'] :\,(\Psi_K)_{ij}>0\}$, i.e., nodes that are directly connected to completely stubborn nodes $[n']$. Since $\sum_{j=n'+1}^n(((\Psi_K)_{22})_{11})_{ij}<1$, $\forall i\in P$ we know that $\rho(((\Psi_K)_{22})_{11})<1$ and $((\Psi_K)_{22})_{22}$ is row-stochastic. However, note $\mathcal{G}$ is connected and thus a node is either completely stubborn or connected to such a node. We conclude $(\Psi_K)_{22}=((\Psi_K)_{22})_{11}$ and hence $\rho((\Psi_K)_{22})=\rho(((\Psi_K)_{22})_{11})<1$. The limit $\lim_{t\rightarrow\infty} (\Psi_K)^t$ hence exists and corresponds to
    \begin{align*}
        \lim_{t\rightarrow\infty}(\Psi_K)^t=\begin{bmatrix}
            I_{n'}&\pmb{0}\\
            (I_{n-n'}-(\Psi_K)_{22})^{-1}\Psi_{21}&\pmb{0}
        \end{bmatrix}.
    \end{align*}
    \item When $n'=0$, we know that $\Psi_K$ is well-defined, as $\Lambda_\alpha\neq I_n$. For $K<\infty$, we expand the recursion and have  $\Psi_K=I_n-\Lambda_\alpha+(\Lambda_\alpha W)^K+\ldots$, which shows it is connected and aperiodic, and hence primitive. When $K=\infty$, we can write $\Psi_\infty=\sum_{i=0}^{\infty}(\Lambda_\alpha W)^i(I_n-\Lambda_\alpha)$ in the following form under node permutation
    \begin{align*}
        \Psi_\infty=\begin{bmatrix}
            (\Psi_\infty)_{11}&\pmb{0}\\
            (\Psi_\infty)_{21}&\pmb{0}
        \end{bmatrix}
    \end{align*}
    where $(\Psi_\infty)_{11}$ is positive and row-stochastic, and this block corresponds to nodes with $\alpha_i<1$. Observe
    \begin{align*}
        \lim_{t\rightarrow\infty}(\Psi_\infty)^t=\begin{bmatrix}
            \lim_{t\rightarrow\infty}((\Psi_\infty)_{11})^t&\pmb{0}\\
            (\Psi_{\infty})_{21}\lim_{t\rightarrow\infty}((\Psi_\infty)_{11})^t&\pmb{0}
            \end{bmatrix}.
    \end{align*}
    where $\lim_{t\rightarrow\infty}((\Psi_\infty)_{11})^t$ exists since $(\Psi_\infty)_{11}$ is primitive.
    \end{itemize}
    We conclude that  $\lim_{t\rightarrow\infty}(\Psi_K)^t$ exists for any $\Lambda_\alpha\neq I_n$. Hence it follows that for any $\pmb{0}\leq\Lambda_\beta\leq\pmb{I}_n$, $\lim_{t\rightarrow\infty}(\Psi_K\Lambda_\beta)^t$ exists. 
    Further,  given the assumption that $\pmb{0}\leq\Lambda_{\beta}\leq I_n$ with $\Lambda_{\beta}\neq I_n$, we have $\rho(\Psi_K \Lambda_{\beta})\leq\rho(\Psi_K)=1$ since $\Psi_K \Lambda_{\beta}\leq \Psi_K$ by applying a corollary of Perron-Frobenius theorem (Corollary 1.5.(a) in~\cite{berman1994nonnegative}). 

    Building on these results, we can now characterize the stable point.
    
    If $\rho(\Psi_K\Lambda_\beta)<1$, the dynamics in \eqref{eq:iterative_update} are Schur stable and the convergence is immediate. 
    We have 
    \[\xv_\mathrm{PS}=(I_n-\Psi_K\Lambda_\beta)^{-1}\Psi_K(I_n-\Lambda_\beta)\xv^*.\]

    Now we discuss the case when $\rho(\Psi_K\Lambda_\beta)=1$. 
    Apply Lemma 2 in \cite{parsegov2016novel}, we know for row-substochastic matrix $\Psi_K\Lambda_\beta$, under node permutation, it can be decomposed into an upper triangular form
    \begin{align*}
        \Psi_K\Lambda_\beta=\begin{bmatrix}
            (\Psi_K\Lambda_\beta)_{11}&(\Psi_K\Lambda_\beta)_{12}\\
            \pmb{0}&(\Psi_K\Lambda_\beta)_{22}
        \end{bmatrix}
    \end{align*}
    where $(\Psi_K\Lambda_\beta)_{11}$ is Schur stable and $(\Psi_K\Lambda_\beta)_{22}$ is row-stochastic. Note $\Psi_K$ is row-stochastic and $0\leq\beta_i\leq 1$ $\forall i\in[n]$. Then for any node $i$ in row-stochastic matrix $(\Psi_K\Lambda_\beta)_{22}$,
    \begin{align*}
        1=\sum_{j}(\Psi_K)_{ij}\beta_j\leq 1.
    \end{align*}
    The equality holds only when for any node $j$ in row-stochastic matrix $(\Psi_K\Lambda_\beta)_{22}$, $\beta_j=1$. Hence $(\Psi_K\Lambda_\beta)_{22}=(\Psi_K)_{22}$. Let $\Psi_{22}^*:=\lim_{t\rightarrow\infty}(\Psi_{22})^t$ that exists since we've proved the existence of $\lim_{t\rightarrow\infty}(\Psi_K)^t$. The equilibrium is cahracterized as follows
    \begin{align}
        \xv_\mathrm{PS}
        &=\begin{bmatrix}
            (I-(\Psi_K\Lambda_\beta)_{11})^{-1}&\pmb{0}\\
            \pmb{0}&I
        \end{bmatrix}\begin{bmatrix}
            (\Psi_K)_{11}-(\Psi_K\Lambda_\beta)_{11}&(\Psi_K)_{12}\Psi_{22}^*\\
            \pmb{0}&\Psi_{22}^*
        \end{bmatrix}\xv^*.
        \label{eq:generalPS}
    \end{align}
    
    In particular, when $\Lambda_{\beta}=I_n$, the block $(\Psi_K\Lambda)_{11}$ is empty, and we have $\xv_\mathrm{PS}=\lim_{t\rightarrow\infty}\Psi_K^t$. We proved earlier that this limit  exists when $\Lambda_\alpha\neq I_n$.

    Now, we prove that the inverse, $(I_n-\Psi_K\Lambda_{\beta})^{-1}$, in the equilibrimu characterization \eqref{eq:generalPS} is well-defined when $\alpha_i>0\;\forall i\in[n]$, $\Lambda_\alpha\neq I_n$, and $\Lambda_\beta\neq I_n$. We apply a corollary of Perron-Frobenius theorem (Corollary 1.5.(b) in~\cite{berman1994nonnegative}) and notice that $\Psi_K\Lambda_{\beta}\leq \Psi_K$, $\Psi_K\Lambda_{\beta}\neq\Psi_K$, and $\Psi_K+\Psi_K\Lambda_{\beta}$ is irreducible since $\alpha_i>0\;\forall i\in[n]$. Hence $\rho(\Psi\Lambda_{\beta})<1$ and we have
    \begin{align}
        \xv_\mathrm{PS}=(I_n-\Psi_K\Lambda_\beta)^{-1}\Psi_K(I_n-\Lambda_\beta)\xv^*.
        \label{eq:ps_chara}
    \end{align}
\end{proof}

\subsection{Proof of  Proposition~\ref{prop:convergence_rate_sl}}
\label{sec:proof-convergence}

\begin{proof}
    For non-negative matrices, we apply a corollary from Perron-Frobenius theorem (Corollary 1.5.(b) in~\cite{berman1994nonnegative}): If $\pmb{0}\leq A\leq B$, then $\rho(A)\leq \rho(B)$. Observe
    \begin{align*}
        &\pmb{0}\leq\Psi_K\leq \max_k\beta_k\Psi_K,\\ 
        &\rho(\Psi_K \Lambda_{\beta})\leq \rho(\max_k\beta_k\pmb{\Psi}_K)=\max_k\beta_k\rho(\Psi_K)=\max_k\beta_k<1.
    \end{align*}
    As for the case when some $\beta_i=1$, the convergence rate and choice of $c$ is immediate from the proof of Proposition~\ref{prop:existencePS}.
\end{proof}

\subsection{Proof of distribution sensitivity}
\label{app:distribution_sensitivity}

\begin{proof}
    Set $\pmb{f}^{(t)}$ to $\pmb{f}_1^{(t)},\pmb{f}_2^{(t)}$ respectively, and let $\pmb{f}_1^{(t+1)},\pmb{f}_2^{(t+1)}$ be the corresponding predictions at $t+1$. We have $\|\pmb{f}_1^{(t+1)}-\pmb{f}_2^{(t+1)}\|=\|\Psi_K\Lambda_\beta(\pmb{f}_1^{(t)}-\pmb{f}_2^{(t)})\|\leq \|\Lambda_\beta\|\|\Psi_K\|\|\pmb{f}_1^{(t)}-\pmb{f}_2^{(t)}\|$ where $\|\Lambda_\beta\|\|\Psi_K\|$ corresponds to the sensitivity parameter in~\citet{PZMH20}.
\end{proof}

\subsection{Proof of Proposition~\ref{prop:mean_variance_sl}}
\label{app:var_peer_platform_sus}

To study the relationship between variance and susceptibilities we build on the following intermediate result decomposing the equilibrium into a consensus component and a sum of heterogeneous components.
\begin{proposition}[Decomposition of equilibrium]
\label{theorem:equilibrium_decomposition}
Consider a network characterized by a regular, simple, and connected graph $\mathcal{G}$, its nodes with peer susceptibilities $\alpha_i\equiv\alpha\in(0,1)$, platform susceptibilities $\beta_i\equiv\beta\in[0,1]$, and a platform that enforces perfect predictions based on policy~\eqref{eq:SL}. Suppose $K=\infty$ and FJ dynamics converges. Then, the equilibrium $\xv_{\mathrm{PS}}$ can be decomposed as follows 
   \begin{align}
        \label{eq:ps_decomposition_homogeneous}
        \pmb{x}_{\mathrm{PS}}&=\frac{\pmb{1}^{\top}\pmb{y}^*}{n(1-\beta)}   \pmb{1} + \sum_{i=2}^n\frac{\lambda_i \pmb{v}_i^{\top}\yv^*}{1-\beta\lambda_i}\pmb{v}_i,
    \end{align}
    where $\yv^*=(1-\beta)\pmb{x}^*$ and $(\lambda_i, \pmb{v}_i)$ are eigenvalues and corresponding eigenvectors of $\Psi_{\infty}$ with $\pmb{v}_i\perp \pmb{1}\;\forall i\geq 2$. 
\end{proposition}
\begin{proof}[Proof of Proposition~\ref{theorem:equilibrium_decomposition}]
    We work with the explicit expression of $\xv_{\mathrm{PS}}$ in~\eqref{eq:ps_chara} where \[\Psi_K\coloneqq\sum_{i=0}^{K-1}(\Da W)^i(I_n-\Da)+(\Da W)^K.\]Let us first focus on the matrix $W$. The matrix is symmetric and we apply the Schur decomposition  $W=U\Lambda U^{\top}$ with $\Lambda = \diag(\mu_1,\cdots,\mu_n)$ and $U\coloneqq[\pmb{v}_1,\cdots,\pmb{v}_n]$ where $\{\mu_i\,\pmb{v}_i\}_{i\in[n]}$ denote the eigenvalues and eigenvectors of $W$. By definition we have $U^{\top}U=UU^{\top}=I_n$. Further, since $W$ is row-stochastic, apply Gershgorin circle theorem, we have $\rho(W)=1$. Apply Perron-Frobenius theorem on irreducible matrix $W$, w.l.o.g., let $\mu_1=1$ and $\pmb{v}_1=\frac{1}{\sqrt{n}}\pmb{1}^{\top}$ and $\mu_i\neq 1$ and $|\mu_i|\leq 1$ $\forall i\geq 2$. Using this expression in the characterization of stability, we obtain 
$\Psi_{\infty}=U\diag(\lambda_1,\cdots,\lambda_n)U^{\top}$ with $\lambda_i\coloneqq \frac{1-\alpha}{1-\alpha\mu_i}$, and hence
\begin{align*}
        \xv_{\mathrm{PS}}&=U\diag\left(\frac{\lambda_1}{1-\beta\lambda_1},\cdots,\frac{\lambda_n}{1-\beta\lambda_n}\right)U^{\top}.
    \end{align*}
    We can easily verify that $\lambda_1=1$, and $0<\lambda_i<1$ $\forall i\geq 2$ by noticing $\mu_i\neq 1$, $|\mu_i|\leq 1$ for $\forall i\geq 2$, and  $1-\alpha\mu_i-1+\alpha=\alpha(1-\mu_i)>0$. This concludes the proof of Proposition~\ref{theorem:equilibrium_decomposition}.
\end{proof}
Before we move to the proof of Proposition~\ref{prop:mean_variance_sl} let us comment on this result. The characterization applies to regular networks and homogeneous parameters. It shows that the equilibrium has an invariant consensus component $\frac{\pmb{1}}{\sqrt{n}}$ and its coefficient  $\frac{1}{1-\beta}$  is invariant with respect to the choice of $\Lambda_{\alpha}$ and $\mathcal{G}$. The other components characterize the heterogeneity of opinion dynamics in different directions.  The more the heterogeneous components $\pmb{v}_i$, $i\geq 2$ are aligned with the invariant component $(1-\beta)\xv^*$ and the larger the eigenvalue $\lambda_i$, the more $\pmb{v}_i$ will contribute to the equilibrium. While we notice that the consensus component ``dominates'' in the sense that
\begin{align*}
    0<\frac{\lambda_i}{1-\beta\lambda_i}<\frac{1}{1-\beta},\,\forall i.
\end{align*}

Next we build on this decomposition to prove Proposition~\ref{prop:mean_variance_sl}.
\begin{proof}[Proof of Proposition~\ref{prop:mean_variance_sl}]
Note $\pmb{v}_i\perp\pmb{1}$ $\forall i\geq 2$, 
\begin{align*}
    \bar{\xv}_{\mathrm{PS}}&=\frac{(1-\beta)\pmb{1}^{\top}\sum_{i=1}^n\frac{\lambda_i}{1-\beta\lambda_i}\pmb{v}_i\pmb{v}_i^{\top}\xv^*}{n}\\
    &=\frac{\pmb{1}^{\top}\xv^*}{n}.
\end{align*}
As for the variance of opinions, we have
\begin{align*}
    \mathrm{Var}(\xv_{\mathrm{PS}})
    &=\left\|\sum_{i=1}^n \frac{\lambda_i}{1-\beta\lambda_i}\pmb{v}_i^{\top}((1-\beta)\pmb{x}^*)\pmb{v}_i - \bar{\xv}_{\mathrm{PS}}\right\|^2\\
    &=\left\|\sum_{i=2}^n \frac{\lambda_i}{1-\beta(1-\eta)\lambda_i}\pmb{v}_i^{\top}((1-\beta)\xv^*)\pmb{v}_i\right\|^2\\
    &=\sum_{i=2}^n\left(\frac{\lambda_i}{1-\beta\lambda_i}\pmb{v}_i^{\top}((1-\beta)\xv^*)\right)^2.
\end{align*}
As for the relationship between platform susceptibility $\beta$ and $\mathrm{Var}(\xv_{\mathrm{PS}})$, we have
\begin{align*}
    \frac{\partial \mathrm{Var}(\xv_{\mathrm{PS}})}{\partial\beta}
    &=\sum_{i=2}^n\frac{2\lambda_i^2(1-\beta)(\lambda_i-1)(\pmb{v}_i^{\top}\pmb{x}^*)^2}{(1-\beta\lambda_i)^3}<0.
\end{align*}
\end{proof}

Similarly we can also show that under the conditions of Proposition~\ref{prop:mean_variance_sl} the variance decreases as $\alpha$ grows. We omit this result form the main body as it follows from classical insights in peer dynamics.  Namely, via basic calculations, we have 
\begin{align*}
    \frac{\partial\left(\frac{\lambda_i}{1-\beta\lambda_i}\right)^2}{\partial\alpha}=\frac{-\lambda_i(1-\mu_i)}{(1-\beta\lambda_i)^3(1-\alpha\mu_i)^2}<0\quad i\geq 2
\end{align*}
and hence if there exists $i_0\geq2$ such that $\pmb{v}_{i_0}^{\top}\xv^*\neq 0$, the larger $\alpha$ is, the larger $\mathrm{Var}(\xv_{\mathrm{PS}})$ is. Also, note $\mathrm{Var}(\xv^*)=\sum_{i=2}^n(\pmb{v}_i^\top\xv^*)^2$, hence $\mathrm{Var}(\xv_\mathrm{PS})$ is also implicitly determined by $\mathrm{Var}(\xv^*)$. It follows that 
\[\frac{\partial \mathrm{Var}(\xv_{\mathrm{PS}})}{\partial\alpha}<0.\]

\subsection{Proof of Theorem~\ref{theorem:consensus}}
\label{app:consensus_sl}

Next we show consensus in the limit $\beta_i\rightarrow 1\;\forall i$. We proof the claim without the homogeneity assumption on $\beta$. Instead assume the weaker numerical condition on how the limit is attained for different individuals: All we assume is that the limit 
$\lim_{\beta_i,\beta_j\rightarrow 1}\frac{1-\beta_i}{1-\beta_j}$ exists.

\begin{proof}
To prove Theorem~\ref{theorem:consensus} we first  show that $\Psi_K$ is primitive under the conditions of the theorem. We have
\begin{align*}
\Psi_1&=(I_n-\Lambda_{\alpha})+\Lambda_{\alpha}W\\
    \Psi_K&=\Lambda_{\alpha}W(I_n-\Lambda_{\alpha})+(I_n-\Lambda_{\alpha})+\sum_{i=2}^{K-1}(\Lambda_{\alpha}W)^i(I_n-\Lambda_{\alpha})+(\Lambda_{\alpha}W)^K,\quad K\geq 2.
\end{align*}
We note that $\Lambda_{\alpha}W(I_n-\Lambda_{\alpha})$ is  irreducible given that $\alpha_i\in(0,1)$ for $\forall i\in[n]$.\footnote{A side note is that for $K<\infty$, $\Psi_K$ is primitive with a weaker assumption that $\Lambda_\alpha\neq I_n$.}
By Lemma~\ref{lemma:primitive_matrix}, we know that $\Psi_K$ is primitive for any $K\in\{\mathbb{N}\cup\infty\}$. 

From the equilibrium characterization in Equation~\eqref{eq:ps_chara}, we have that 
for $\Lambda_{\beta}=\pmb{1}$ it holds that $\xv_\mathrm{PS}=\lim_{t\rightarrow\infty}\Psi_K^t\pmb{x}^*.$ 
Hence, we get convergence to  
\[c^*=\pmb{y}\pmb{f}^{(0)}=\pmb{y}^{\top}\pmb{x}^*\] where $\pmb{y}$ is the eigenvector of $(\Psi_K)^{\top}$ and satisfies $\pmb{1}^{\top}\pmb{y}=1$.

It remains to show that in the limit  $\Lambda_{\beta}\rightarrow\pmb{1}$ this behavior is maintained. 
Thus, for any sequence $\Lambda_{\beta}\rightarrow\pmb{1}$, we choose $\epsilon>0$ and an arbitrary $c>1$ such that $\frac{\epsilon}{c}\preceq I_n-\Lambda_{\beta}\preceq \epsilon I_n$ 
\begin{align*}
    \left(\sum_{m=0}^{\infty}(\Psi_K(I_n-\epsilon I_n))^m\frac{\epsilon}{c}\right)_{ij}&\leq \left(\sum_{m=0}^{\infty}(\Psi_K\Lambda_{\beta})^m(I_n-\Lambda_{\beta})\right)_{ij}\\
    &\leq \left(\sum_{m=0}^{\infty}(\Psi_K(I_n-\frac{\epsilon}{c}I_n))^m\epsilon\right)_{ij}
\end{align*} 
where the inequality holds since $\Psi$ is non-negative and $I_n-\Lambda_{\beta}$ is diagonal with positive entries. 
Under the assumption that $\lim_{i,j}\frac{1-\beta_i}{1-\beta_j}$ exists, let~$\epsilon\rightarrow 0$ and apply the theorem from~\cite{friedkin1990social}, we get
\begin{align*}
    \left(\frac{1}{c}\lim_{m\rightarrow\infty}(\Psi_K)^m\right)_{ij}&\leq \lim_{\beta\rightarrow\pmb{1}}\left(\sum_{m=0}^{\infty}(\Psi_K\Lambda_{\beta})^m(I_n-\Lambda_{\beta})\right)_{ij}\\
    &\leq \left(c\lim_{m\rightarrow\infty}(\Psi_K)^m\right)_{ij}.
\end{align*}
Since $c>1$ can be chosen arbitrarily, apply equilibrium~\eqref{eq:ps_chara} we have
\begin{align*}
    \lim_{\Lambda_{\beta}\rightarrow\pmb{1}}(I_n-\Psi_K\Lambda_{\beta})^{-1}(I_n-\Lambda_{\beta})=\lim_{j\rightarrow\infty}(\Psi_K)^j.
\end{align*}
Apply Corollary~\ref{corollary:perron_frobenius_coro}, opinions reach consensus as follows
\begin{align*}
    \lim_{\Lambda_{\beta}\rightarrow\pmb{1}}\pmb{x}_{\mathrm{PS}}&=\pmb{1}^{\top}\pmb{y}\lim_{\Lambda_{\beta}\rightarrow\pmb{1}}(I_n-\Lambda_{\beta})^{-1}\Psi_K(I_n-\Lambda_{\beta})\pmb{x}^*.
\end{align*}
\end{proof}

\subsection{Proof of Proposition~\ref{prop:consensus_mean_estimation}}

\begin{proof} 
Throughout this section, let $n':=|\mathcal{V}'|$ and define
\begin{align*}
\Lambda_{\beta,1}&:=(\Lambda_{\beta})_{1:n', 1:n'} 
&\Lambda_{\beta,2}&:=(\Lambda_{\beta})_{n'+1:n, n'+1:n},\\ 
\xv_\mathrm{PS}'&:=(\xv_\mathrm{PS})_{1:n'}
&\Psi_K'&:=(\Psi_K)_{1:n',1:n'}.
\end{align*}
Without loss of generality we can rewrite the predictions as follows,
   \begin{align*}
       \pmb{f}^{(t+1)}&= M\xv_\mathrm{ex}^{(t)},\quad M\coloneqq \begin{bmatrix}
           I_{n'}&\pmb{0}_{n'\times (n-n')}\\
           \frac{1}{n'}\pmb{1}_{(n-n')\times n'}&\pmb{0}_{(n-n')\times(n-n')}
       \end{bmatrix}.
   \end{align*}
   Hence we have
   \begin{align*}
       \xv_\mathrm{ex}^{(t+1)}&=\Psi_K((I_n-\Lambda_\beta)\xv^*+\Lambda_\beta M\xv_\mathrm{ex}^{(t)}).
   \end{align*}
    We look at $\Psi_K\Lambda_\beta M$ in a block-wise way,
    \begin{align*}
        \Psi_{K}\Lambda_\beta M&=\begin{bmatrix}
            (\Psi_{K}\Lambda_\beta M)_{1:n', 1:n'}&\pmb{0}_{n'\times (n-n')}\\
            *&\pmb{0}_{(n-n')\times(n-n')}
        \end{bmatrix},
    \end{align*}
    \begin{itemize}
        \item When $\Lambda_{\beta}=I_n$, we have
    \begin{align*}
        \xv_\mathrm{ex}^{(t)}=(\Psi_{K}M)^t\xv_\mathrm{ex}^{(0)}.
    \end{align*}
    It is easy to verify that $\Psi_{K}M$ is still a row-stochastic matrix and 
    \begin{align*}
        \Psi_{K}M&=\begin{bmatrix}
            (\Psi_{K}M)_{1:n', 1:n'}&\pmb{0}_{n'\times (n-n')}\\
            *&\pmb{0}_{(n-n')\times(n-n')}
        \end{bmatrix},\\(\Psi_{K}M)^t&=\begin{bmatrix}
            ((\Psi_K M)_{1:n', 1:n'})^t&\pmb{0}_{n'\times (n-n')}\\
            *&\pmb{0}_{(n-n')\times(n-n')}
        \end{bmatrix}
    \end{align*}
    The primitivity of the block $(\Psi_{K}M)_{1:n', 1:n'}$ comes from an observation: We assume a virtual node $n+1$ assigned with the opinions of $\mathcal{V}'$, which represents the union of $\mathcal{V}\backslash\mathcal{V}'$. From each node in $\mathcal{V}'$, there's a directed edge to node $n+1$. While the node $n+1$ is connected with some of nodes in $\mathcal{V}'$ with edges starting at $n+1$ because of the connectivity assumption. Hence for any pair of node in $\mathcal{V}'$, there is a path between them, the primitivity is immediate by noticing $(\Psi_K)_{ii}>0$ for $i\in[n]$. Hence opinions in $\mathcal{V}'$ will converge to consensus using the same proof as Theorem~\ref{theorem:consensus}. As the platform is using the mean predictor, we conclude that the whole population will converge to consensus as $t\rightarrow \infty$. We observe that the consensus opinion $c'$ is influenced by nodes in $\mathcal{V}'$ and no assumption about $\alpha_j$,~$j\in\mathcal{V}/\mathcal{V}'$, is needed. And when $\pmb{x}^*$ differs from $\tilde{\xv}^*$ in some $i\in\mathcal{V}'$, the consensus also differs since $\lim_{t\rightarrow\infty}((\Psi_K M)_{1:n',1:n'})^t$ is a positive matrix. 
    \item For general $\Lambda_{\pmb{\beta}}$, since $\lim_{t\rightarrow\infty}((\Psi_K M)_{1:n',1:n'})^t$ exists, we know $\lim_{t\rightarrow\infty}((\Psi_K\Lambda_\beta M)_{1:n',1:n'})^t$ also exists and apply the proof of Proposition~\ref{prop:existencePS}. The dynamics over $\mathcal{V}'$ converges. And the equilibrium can be characterized as follows
    \begin{align*}
        \xv_\mathrm{PS}&=\Psi_K\begin{bmatrix}
            \Lambda_{\pmb{\beta},1}\xv_\mathrm{PS}'+(I_{n'}-\Lambda_{\beta,1})(\xv^*)_{1:n'}\\
            \frac{\sum_{i\in\mathcal{V}'}(\xv_\mathrm{PS})_i}{n'}\Lambda_{\beta,2}\pmb{1}+(I_{n-n'}-\Lambda_{\beta,2})(\xv^*)_{n'+1:n}
        \end{bmatrix},
    \end{align*}
    where $\xv_\mathrm{PS}'$ is influenced by $x_i^*\;\forall i\in\mathcal{V}'$ as long as the coefficient of $\frac{\sum_{i\in\mathcal{V}'}(\xv_\mathrm{PS})_i}{n'}$ is positive. Explicitly,
    \begin{align*}
        \xv_\mathrm{PS}'&=(\Psi_K\Lambda_\beta M)_{1:n',1:n'}\xv_\mathrm{PS}'+(\Psi_K(I_n-\Lambda_\beta)\xv^*)_{1:n'}\\
        (\xv_\mathrm{PS}\mid_{\xv^*}-\xv_\mathrm{PS}\mid_{\xv^*})&=(x_{i_0}^*-\tilde{x}_{i_0}^*)(I_{n'}-(\Psi_K\Lambda_\beta M)_{1:n',1:n'})^{-1}(\Psi_K(I_n-\Lambda_\beta)\ev_{i_0})_{1:n'}.
    \end{align*}
    Note that since $(I_{n'}-(\Psi_K\Lambda_\beta M)_{1:n',1:n'})^{-1}=\sum_{t=0}^{\infty}((\Psi_K\Lambda_\beta M)_{1:n',1:n'})^t$ is positive. The variation of $\xv_\mathrm{PS}'$ has the same sign as $x_{i_0}^*-\tilde{x}_{i_0}^*$.
    \end{itemize}
    Since $\xv_{i_0}^*\neq\tilde\xv_{i_0}^*$, for some $i_0\in [n']$, $\xv_i^*=\tilde\xv_i^*$, $\forall i\in[n']/\{i_o\}$, and note $\alpha_q=0$, we have
    \begin{align*}
        \xv_\mathrm{PS}\mid_{\xv^*}-\xv_\mathrm{PS}\mid_{\tilde\xv^*}&=\Psi_K\begin{bmatrix}
            \Lambda_{\beta,1}(\xv_\mathrm{PS}'\mid_{\xv^*}-\xv_\mathrm{PS}'\mid_{\tilde{\xv}^*})+(1-\beta_{i_0})(x_{i_0}^*-\tilde{x}_{i_0}^*)\ev_{i_0}\\
            \frac{\sum_{i\in\mathcal{V}'}((\xv_\mathrm{PS}\mid_{\xv^*})_i-\xv_\mathrm{PS}\mid_{\tilde{\xv}^*})_i)}{n'}\Lambda_{\beta,2}\pmb{1}
        \end{bmatrix}\\
       (\xv_\mathrm{PS}\mid_{\tilde\xv})_q-(\xv_\mathrm{PS}\mid_{\tilde\xv^*})_q&=\sum_{i=1}^{n'}(\Psi_K)_{q,i}\beta_i((\xv_\mathrm{PS}\mid_{\xv^*})_i-(\xv_\mathrm{PS}\mid_{\tilde\xv^*})_i)\\
       &\quad+\sum_{k=1}^{n'}\frac{\sum_{i\in[n']}(\xv_\mathrm{PS}\mid_{\xv^*})_i-(\xv_\mathrm{PS}\mid_{\tilde\xv^*})_i}{n'}\sum_{j=n'+1}^n(\Psi_K)_{q,j}\beta_j\\
       &> \beta_q\sum_{k=1}^{n'}\frac{\sum_{i\in[n']}(\xv_\mathrm{PS}\mid_{\xv^*})_i-(\xv_\mathrm{PS}\mid_{\tilde\xv^*})_i}{n'}\\
       &=\frac{\beta_q(x_{i_0}^*-\tilde{x}_{i_0}^*)}{n'}\pmb{1}^\top (I_{n'}-(\Psi_K\Lambda_\beta M)_{1:n',1:n'})^{-1}(\Psi_K(I_n-\Lambda_\beta)\ev_{i_0})_{1:n'}.
    \end{align*}
    Specially when $\Lambda_{\beta,1}=I_{n'}$,
    \begin{align*}
        (\xv_\mathrm{PS}\mid_{\tilde\xv})_q-(\xv_\mathrm{PS}\mid_{\tilde\xv^*})_q&=((\xv_\mathrm{PS}\mid_{\xv^*})_1-(\xv_\mathrm{PS}\mid_{\tilde{\xv}^*})_1)(1-\sum_{j=n'+1}^{n}(1-\beta_j)(\Psi_K)_{q,j})\\
        &=y_{i_0}'(x_{i_0}^*-\tilde{x}_{i_0}^*)(1-\sum_{j=n'+1}^{n}(1-\beta_j)(\Psi_K)_{q,j}),
    \end{align*}
    where $\pmb{y}'=\begin{bmatrix}
        y_1',\cdots,y_{n'}'
    \end{bmatrix}^\top\in\mathbb{R}^{n'}$ is the eigenvector of $(\Psi_KM)_{1:n',1:n'}^\top$ w.r.t. 1. Since $(\Psi_KM)_{1:n',1:n'}^\top$ is irreducible, apply Theorem (2.10) in \citet{berman1994nonnegative}, we know $\pmb{y}'$ is positive. Hence by setting $c:=\min_{i\in[n']} |y_i'|$, we conclude the proof.
\end{proof}

\subsection{Proof of Theorem~\ref{prop:spillover}}
\label{app:spillover}
To characterize the indirect platform influence, we mainly prove and utilize the fact that the equilibrium under the assumptions of Theorem~\ref{prop:spillover} can be written as
\begin{align}
    \xv_\mathrm{PS} = \frac{\beta s(1-\gamma+\gamma\psi_2)\psi_1}{1-\gamma+2\gamma\psi_2-\gamma^2\psi_1\psi_2}e_j-\frac{\beta s\gamma\psi_1\psi_2}{1-\gamma+2\gamma\psi_2-\gamma^2\psi_1\psi_2}e_l+\frac{\beta s\psi_2}{1-\gamma+2\gamma\psi_2-\gamma^2\psi_1\psi_2}\pmb{1}
    \label{eq:x_ps_decomposition}
\end{align}
where $\psi_1:=\frac{(1-\alpha)(n-1)}{n-1+\alpha}>0$, $\psi_2:=\frac{\alpha}{n-1+\alpha}>0$.

With the decomposition in~\eqref{eq:x_ps_decomposition}, we have
\begin{align*}
    \mathrm{Mean}(\xv_\mathrm{PS})&=\frac{\pmb{1}^\top \xv_\mathrm{PS}}{n}=\frac{\beta s(1-\gamma\psi_1)}{n(1-\gamma+2\gamma\psi_2-\gamma^2\psi_1\psi_2)}.
    \label{eq:PSapp}
\end{align*}
In particular, $(\xv_\mathrm{PS})_l=\frac{\beta s\psi_2(-\gamma\psi_1+1)}{1-\gamma+2\gamma\psi_2-\gamma^2\psi_1\psi_2}$.
    Take the derivative of $\mathrm{Mean}(\xv_\mathrm{PS})$, $(\xv_\mathrm{PS})_l$ with respect to $\gamma$ and we have 
    \begin{align*}
        \frac{\partial \mathrm{Mean}(\xv_\mathrm{PS}) }{\partial \gamma}&=\frac{\beta s(1-\psi_1-2\psi_2+2\gamma\psi_1\psi_2-\gamma^2\psi_1^2\psi_2)}{(1-\gamma+2\gamma\psi_2-\gamma^2\psi_1\psi_2)^2}> 0.\\
        \frac{\partial (\xv_\mathrm{PS})_l}{\partial \gamma}&=\frac{\beta s \psi_2(1-\psi_1-2\psi_2+2\psi_1\psi_2-\gamma^2\psi_1^2\psi_2)}{(1-\gamma+2\gamma\psi_2-\gamma^2\psi_1\psi_2)^2}>0.
    \end{align*}
    The above inequalities hold since
    \begin{align*}
        1-\psi_1-2\psi_2&\geq 1-\psi_1-n\psi_2=0,\,\text{and}\\
       2\psi_1\psi_2-\gamma^2\psi_1^2\psi_2&> 2\gamma\psi_1\psi_2-\gamma^2\psi_1^2\psi_2>\gamma^2\psi_1\psi_2-\gamma^2\psi_1^2\psi_2\geq 0.
    \end{align*}
Hence $\frac{\partial\left( (\xv_\mathrm{PS})_l-(\xv_\mathrm{PS})_l^{\gamma=0}\right)}{\partial\gamma}>0.$

\begin{proof}[Proof of Equation~\eqref{eq:x_ps_decomposition}]
Now we prove the decomposition of $\xv_\mathrm{PS}$ step by step. With the dynamics described in Proposition~\ref{prop:spillover}, we have the equations
   \begin{align*}
    \xv_\mathrm{init}^{(t+1)}&=(I_n-\Lambda')\xv^*+\tilde{\Lambda}\xv_\mathrm{ex}^{(t)}+\beta s e_j,  \\
       \xv_\mathrm{ex}^{(t+1)}&=\Psi_{\infty}\xv_\mathrm{init}^{(t+1)},
   \end{align*}
   where $\Lambda':=\gamma I_n-\gamma e_le_l^\top + (\beta-\gamma)e_je_j^\top$ and $\tilde{\Lambda}:=\gamma I_n-\gamma e_le_l^\top -\gamma e_je_j^\top$. Let $S:=I_n-\Psi_\infty \Lambda_\beta''$, $\Lambda'':=\gamma I_n-\gamma e_le_l^\top$. Since $\alpha,\gamma,\beta\in(0,1)$, the equilibrium is well-defined:
   \begin{align*}
       \xv_\mathrm{PS}&=(I_n-\Psi_\infty\tilde{\Lambda})^{-1}\Psi_\infty ((I_n-\Lambda')\xv^*+\beta s e_j)\\
       &=\left(S^{-1}-\frac{\gamma S^{-1}\Psi_\infty e_je_j^\top S^{-1}}{1+\gamma e_j^\top S^{-1}\Psi_\infty e_j}\right) \Psi_\infty \left( (I_n-\Lambda''-(\beta-\gamma)e_je_j^\top)\xv^*+\beta s e_j \right)\\
       &=S^{-1}\Psi(I_n-\Lambda'')\xv^*-\frac{\gamma S^{-1}\Psi_\infty e_je_j^\top S^{-1}\Psi_\infty (I_n-\Lambda'')\xv^* }{1+\gamma e_j^\top S^{-1}\Psi_\infty e_j}\\
       &\quad+\left(S^{-1}-\frac{\gamma S^{-1}\Psi_\infty e_je_j^\top S^{-1}}{1+\gamma e_j^\top S^{-1}\Psi_\infty e_j}\right)\Psi_\infty (\beta s -(\beta-\gamma)e_j^\top\xv^*)e_j\\
       &=\beta s \left(S^{-1}-\frac{\gamma S^{-1}\Psi_\infty e_je_j^\top S^{-1}}{1+\gamma e_j^\top S^{-1}\Psi_\infty e_j}\right)\Psi_\infty e_j
   \end{align*}
   where we apply Sherman–Morrison formula in the second equation. Furthermore, we apply Sherman–Morrison formula again and observe 
   \begin{align*}
       S^{-1}&=(I_n-\gamma \Psi_\infty )^{-1}-\frac{\gamma(I_n-\gamma\Psi_\infty)^{-1}\Psi_\infty e_le_l^\top (I_n-\gamma\Psi_\infty)^{-1}}{1+\gamma e_l^\top (I_n-\gamma\Psi_\infty)^{-1}\Psi_\infty e_l}.
   \end{align*}
   Note 
   \begin{align*}
       \Psi_\infty &= (1-\alpha)(I_n-\frac{\alpha}{n-1}(\pmb{1}\pmb{1}^\top-I_n))^{-1}\\
       &=\frac{(1-\alpha)(n-1)}{n-1+\alpha}I_n+\frac{\alpha}{n-1+\alpha}\pmb{1}\pmb{1}^\top\\
       &=\psi_1 I_n+\psi_2\pmb{1}\pmb{1}^\top
   \end{align*}
   where $\psi_1:=\frac{(1-\alpha)(n-1)}{n-1+\alpha}$, $\psi_2:=\frac{\alpha}{n-1+\alpha}$, and $\psi_1+n\psi_2=1$. Hence
   \begin{align*}
       (I_n-\gamma\Psi_\infty)^{-1}
       =\frac{1}{1-\gamma\psi_1}I_n+\frac{\gamma\psi_2}{(1-\gamma\psi_1)(1-\gamma)}\pmb{1}\pmb{1}^\top
       =i_1I_n+i_2\pmb{1}\pmb{1}^\top
   \end{align*}
    where $i_1:=\frac{1}{1-\gamma\psi_1}$, $i_2:=\frac{\gamma\psi_2}{(1-\gamma\psi_1)(1-\gamma)}$, and note $1+\gamma i_1\psi_1=i_1$, $\gamma\psi_2 i_1=(1-\gamma)i_2$. Then we have
    \begin{align*}
        S^{-1}&=i_1I_n+\left(i_2-\frac{\gamma (i_2^2+i_1i_2\psi_2)}{1+\gamma(i_1\psi_1+i_1\psi_2+i_2)}\right)\pmb{1}\pmb{1}^\top  -\frac{\gamma(\psi_1 i_1^2 e_le_l^\top+i_1(i_2+i_1\psi_2)\pmb{1}e_l^\top+i_1i_2\psi_1 e_l\pmb{1}^\top)}{1+\gamma(i_1\psi_1+i_1\psi_2+i_2)}\\
        &=i_1I_n+a_2\pmb{1}\pmb{1}^\top+a_3 e_le_l^\top +a_4\pmb{1}e_l^\top+a_5 e_l\pmb{1}^\top, 
    \end{align*}
    where
    \begin{align*}
        &a_2:=i_2-\frac{\gamma (i_2^2+i_1i_2\psi_2)}{1+\gamma(i_1\psi_1+i_1\psi_2+i_2)}=\frac{i_2(1+\gamma i_1\psi_1)}{1+\gamma(i_1\psi_1+i_1\psi_2+i_2)}=\frac{i_2 i_1}{i_1+i_2},\\
        &a_3:=\frac{-\gamma\psi_1 i_1^2}{1+\gamma(i_1\psi_1+i_1\psi_2+i_2)}=\frac{i_1-i_1^2}{i_1+i_2},\quad a_4:=\frac{-\gamma i_1 (i_2+i_1\psi_2)}{1+\gamma(i_1\psi_1+i_1\psi_2+i_2)}=\frac{-i_1i_2}{i_1+i_2},\\
        &a_5:=\frac{-\gamma i_1i_2\psi_1}{i_1+i_2}=\frac{i_2-i_1i_2}{i_1+i_2}.
    \end{align*}
    Note $j\neq l$ and then 
    \begin{align*}
        \xv_\mathrm{PS}&=\frac{\beta s}{1+\gamma(i_1\psi_1+a_2+a_4\psi_1+i_1\psi_2+a_4\psi_2)}(i_1\psi_1 e_j+(a_3\psi_2+a_5)e_l+(a_2+i_1\psi_2+a_4\psi_2)\pmb{1})\\
        &=\frac{\beta s(i_1+i_2)}{\gamma(i_1^2+i_1i_2+i_2)}\left((i_1-1)e_j+\frac{i_2-i_1i_2}{(i_1+i_2)}e_l+\frac{i_1i_2}{(i_1+i_2)}\pmb{1}\right),\\
        &=\frac{\beta s(1-\gamma+\gamma\psi_2)\psi_1}{1-\gamma+2\gamma\psi_2-\gamma^2\psi_1\psi_2}e_j-\frac{\beta s\gamma\psi_1\psi_2}{1-\gamma+2\gamma\psi_2-\gamma^2\psi_1\psi_2}e_l+\frac{\beta s\psi_2}{1-\gamma+2\gamma\psi_2-\gamma^2\psi_1\psi_2}\pmb{1}.
    \end{align*}
    which concludes the proof.
\end{proof}

\section{Peer-platform co-influence under the  DeGroot model}
\label{app:degroot}
For the main body of this paper we have focused on the Friedkin-Johnsen model to characterize peer dynamics. However, our analysis can also be extended to the DeGroot model~\citep{degroot1974reaching}. In the following we show how the consensus result and how the intuition form the Friedkin-Johnsen model transfers.  The technical challenge is resolved by the additional assumption that graph $\mathcal{G}$ is primitive. The other resultss can also be extended with minor modifications.

\begin{proposition}[Consensus under DeGroot model]
\label{prop:consensus_degroot}
Consider a network characterized by a simple primitive graph $\mathcal{G}$, its nodes with peer susceptibilities $\Lambda_\alpha=I_n$ and platform susceptibilities $\beta_i\in[0,1]\;\forall i\in[n]$. Suppose $K=\infty$ and the platform performs perfect prediction of past data. Then, there exists a constant $d^*\in[0,1]$ such that a consensus is reached at equilibrium, i.e., 
\[(\xv_\mathrm{PS})_i=d^*, \quad\forall i\in[n].\] Furthermore, if $W$ is doubly-stochastic, we have $d^*=\frac{\sum_{i=1}^n(1-\beta_i)x_i^*}{n-\sum_{i=1}^n\beta_i}$ when $\Lambda_\beta\neq I_n$, and $d^*=\frac{\sum_{i=1}^nx_i^*}{n}$ when $\Lambda_\beta=I_n$.
\end{proposition}
\begin{proof}
    Under supervised learning policy~\eqref{eq:SL}, apply \eqref{eq:sl_iterative_update} and set $\Lambda_\alpha=I_n$, we have
    \begin{align*}
        \pmb{x}_\mathrm{ex}^{(t+1)}=W^{\infty}((I_n-\Lambda_\beta)\xv^*+\Lambda_\beta\pmb{x}_\mathrm{ex}^{(t)}).
    \end{align*}
    \begin{itemize}
        \item When $\Lambda_\beta\neq I_n$, note $W^\infty\Lambda_\beta\leq W^\infty$, $W^\infty\Lambda_\beta\neq W^\infty$ and $W^{\infty}\Lambda_\beta+W^\infty$ is irreducible since $W$ is primitive, apply a corollary of Perron-Frobenius theorem (Corollary 1.5.(b) in~\cite{berman1994nonnegative}) and we have $\rho(W^\infty\Lambda_\beta)<\rho(W^\infty)=1$. Hence we characterize the equilibrium as follows
    \begin{align*}
        \xv_\mathrm{PS}=(I_n-W^{\infty}\Lambda_\beta)^{-1}W^\infty(I_n-\Lambda_\beta)\xv^*=\sum_{t=0}^{\infty}(W^{\infty}\Lambda_\beta)^t W^{\infty}(I_n-\Lambda_\beta)\xv^*.
    \end{align*}
    Apply Perron-Frobenius theorem, we have $(I_n-W^{\infty}\Lambda_\beta)^{-1}$ is well-defined since $\beta_i\in(0,1)$ and $W^\infty$ is primitive. Moreover, we know $(W^\infty\Lambda_\beta)^tW^\infty(I_n-\Lambda_\beta)\xv^*=c_t\pmb{1}$ for some constant $c_t\in[0,1],\forall t\geq 0$. Hence there exists $d^*\in[0,1]$ such that $(\xv_\mathrm{PS})_i=d^*$, $\forall i\in[n]$. 
    \item When $\Lambda_\beta=I_n$, we have consensus again since $W$ is primitive and
    \begin{align*}
        \xv_\mathrm{PS}=W^\infty\xv^*.
    \end{align*}
    \end{itemize}

    When $W$ is doubly-stochastic, we know each node has the same long-run influence by applying Lemma~\ref{corollary:perron_frobenius_coro}. In other words, $\bar{\xv}_\mathrm{init}^{(t)}=\bar{\xv}_\mathrm{ex}^{(t)}$, i.e., when $\Lambda_\beta\neq I_n$,
    \begin{align*}
        nd^*=\sum_{i=1}^n(1-\beta_i)x_i^*+d^*\sum_{j=1}^n\beta_j.
    \end{align*}
    When $\Lambda_\beta=I_n$, we have
    \begin{align*}
        d^*=\frac{\sum_{i=1}^nx_i^*}{n},
    \end{align*}
    and we conclude our proof.
\end{proof}

\section{Results with finite $K$}
\label{app:finite_k}
We consider the following steering policy~\eqref{eq:algo} where only finite $K$ steps are observed for updating policies and the following equilibrium point performative stability
\begin{align}
    \xv_{\mathrm{PS}}\coloneqq \lim_{t\rightarrow\infty}\xv_\mathrm{ex}^{(t)}.
    \label{eq:ps_k_finite}
\end{align}

\begin{theorem}
\label{theorem:equilibrium_decomposition_finiteK}
    (Decomposition of equilibrium under steering with finite $K$) For connected simple graph $\mathcal{G}$, let $d_i\equiv d\in\mathbb{N}^*$, $\alpha_i\equiv\alpha\in(0,1)$, $\beta_i\equiv\beta\in[0,1]$ for $\forall i\in[n]$. Under \eqref{eq:steering}, the equilibrium opinion, $\pmb{x}_{\mathrm{PS}}$ can be decomposed into consensus and heterogeneous components as follows
    \begin{align*}
        \pmb{x}_{\mathrm{PS}}&=\sum_{i=1}^n\frac{\lambda_i}{1-\beta\lambda_i}\pmb{v}_i\pmb{v}_i^{\top}((1-\beta)\pmb{x}^*),
    \end{align*}
    where $\{\lambda_i\}$ are eigenvalues of $\Psi_K$ with $\{\pmb{v}_i\}$ being the corresponding eigenvectors. W.l.o.g., let $\lambda_1=1$, $\pmb{v}_1=\frac{1}{\sqrt{n}}\pmb{1}^{\top}$ and we have 
    \begin{align*}
        \frac{\lambda_1}{1-\beta\lambda_1}&=\frac{1}{1-\beta},\\ |\frac{\lambda_i}{1-\beta\lambda_i}|&<\frac{1}{1-\beta},\,\forall i\geq 2.
    \end{align*}
\end{theorem}
    \begin{proof}
        First of all, since $W$ is symmetric, $W$ can be decomposed as follows
        \begin{align*}
            W&=U\diag(\mu_1,\cdots,\mu_n)U^{\top},\,U\coloneqq\begin{bmatrix}
                \pmb{v}_1&\cdots&\pmb{v}_n
            \end{bmatrix},
        \end{align*}
        where $U^{\top}U=UU^{\top}=I_n$ and $\{\mu_i\}_{i\in[n]}$ are eigenvalues of $W$ and $\{\pmb{v}_i\}_{i\in[n]}$ are corresponding eigenvectors. Since $W$ is row-stochastic, apply Gershgorin circle theorem, we have $\rho(W)=1$. Apply Perron-Frobenius theorem since $W$ is also irreducible, we obtain $\mu_1=1$ and $\pmb{v}_1=\frac{1}{\sqrt{n}}\pmb{1}^{\top}$ w.l.o.g. and $\mu_i\neq 1$ and $|\mu_i|\leq 1$ for $\forall i\geq 2$. From \eqref{eq:ps_chara}, we can easily obtain
        \begin{align*}
            \Psi_K&=U\diag(\lambda_1,\cdots,\lambda_n)U^{\top},\\
            \lambda_i&\coloneqq \frac{1-\alpha}{1-\alpha\mu_i}+\left(1-\frac{1-\alpha}{1-\alpha\mu_i}\right)\alpha^K\mu_i^K,\\
            \xv_{\mathrm{PS}}&=U\diag\left(\frac{\lambda_1}{1-\beta\lambda_1},\cdots,\frac{\lambda_n}{1-\beta\lambda_n}\right)U^{\top}.
        \end{align*}
        We can easily verify that $\lambda_1=1$, $|\lambda_i|<1$ for $\forall i\in[n]$ by noticing $\mu_i\neq 1$, $|\mu_i|\leq 1$ for $\forall i\geq 2$, and observe
        \begin{align*}
            \frac{1-\alpha}{1-\alpha\mu_i}<1,\,
            |\lambda_i|<\max\{1,\alpha^K\mu_i^K\} =1.
        \end{align*}
        Hence we conclude the proof.
    \end{proof}

\begin{lemma}
(Mean and variance of $\pmb{x}_{\mathrm{PS}}$ with finite $K$)
Under the same setting as Theorem~\ref{theorem:equilibrium_decomposition}, and consider the mean and the variance of $\xv_{\mathrm{PS}}$ defined in~\eqref{eq:ps_k_finite}, the following claims apply.
\begin{enumerate}
    \item The invariant opinion $\xv^*$ determine the mean 
    \begin{align*}
        \bar{\xv}_{\mathrm{PS}}=\frac{\pmb{1}^{\top}\xv^*}{n}.
    \end{align*}
    \item The variance $\mathrm{Var}(\xv_{\mathrm{PS}})$ increases as $\alpha$ increases if $K$ is large enough. As for dependence over $\beta$, we have the variance decreases as $\beta$ increases.
\end{enumerate}
\end{lemma}

\begin{proof}
The structure of proof of Proposition~\ref{prop:mean_variance_sl} can be applied here. The only difference in the result is the dependence of variance over peer susceptibilities $\Lambda_{\alpha}$. Notice $\lambda_i$ can be negative. Hence by assuming $K$ is large enough,  
\begin{align*}
    K&\geq \max_i\left\{\log\frac{|1-\alpha|}{\alpha(1-\mu_i)}/\log(\frac{1}{\alpha|\mu_i|})\right\}
\end{align*}
we ensure $\lambda_i>0$ and the claims still apply. 
\end{proof}

\section{Additional simulations}
\label{app:more_simulations}
We defer part of the simulation results to this section. In Figure~\ref{fig:consensus_sl_pokec} and Figure~\ref{fig:mean_steer_pokec}, the same procedure for generating plots with Yelp dataset is applied to Pokec dataset. Note that we use $T=50$ for generating Figure~\ref{fig:platform_variance_sus_pokec} due to computation efficiency. The finite approximation leads to the slight deviation of $\mathrm{Var}(\xv_\mathrm{PS})$ from zero in the case when $\beta=1$. In Figure~\ref{fig:peer_variance_sl}, we use homogeneous peer susceptibilities, $\alpha_i\equiv\alpha$, and vary them while holding platform susceptibilities fixed. We observe the variance of opinion equilibrium decreases as peer susceptibility increases, which was proved for regular networks in Appendix~\ref{app:var_peer_platform_sus}. 

In Figure~\ref{fig:steer2}, under the mean estimation, the result illustrate Proposition~\ref{prop:consensus_mean_estimation} where the (peer)-stubborn individual $l\in U$ converges to a different equilibrium with different opinions in observed set $O$.

\begin{figure*}[t]
\centering
        \begin{subfigure}[b]{0.45\linewidth}
            \centering
    \includegraphics[width=\linewidth]{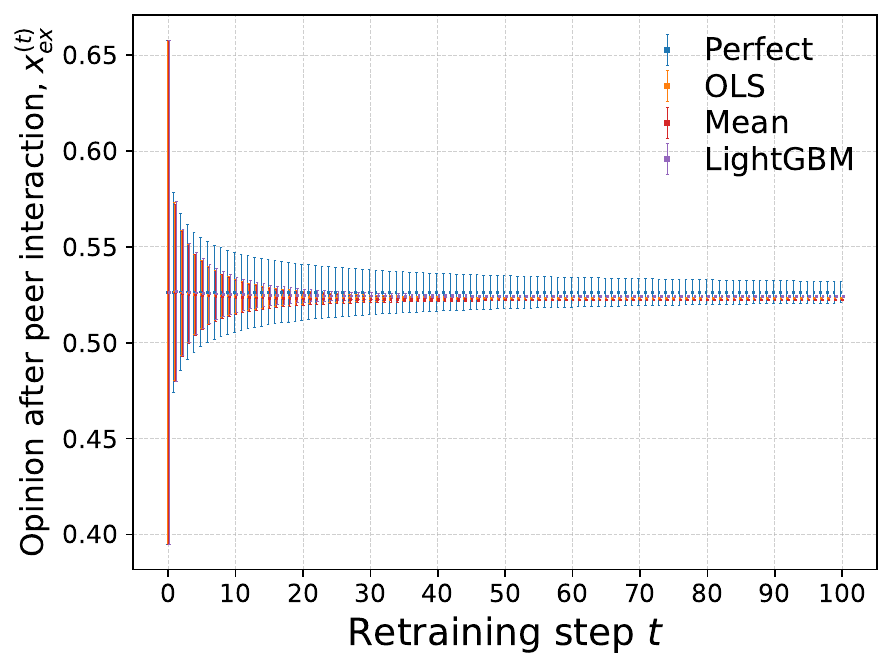}
            \caption{Opinions across retraining steps (Pokec)}
            \label{fig:sl_retrain_steps_pokec}
        \end{subfigure}
        \hfill
            \begin{subfigure}[b]{0.45\linewidth}
            \centering
            \includegraphics[width=\linewidth]{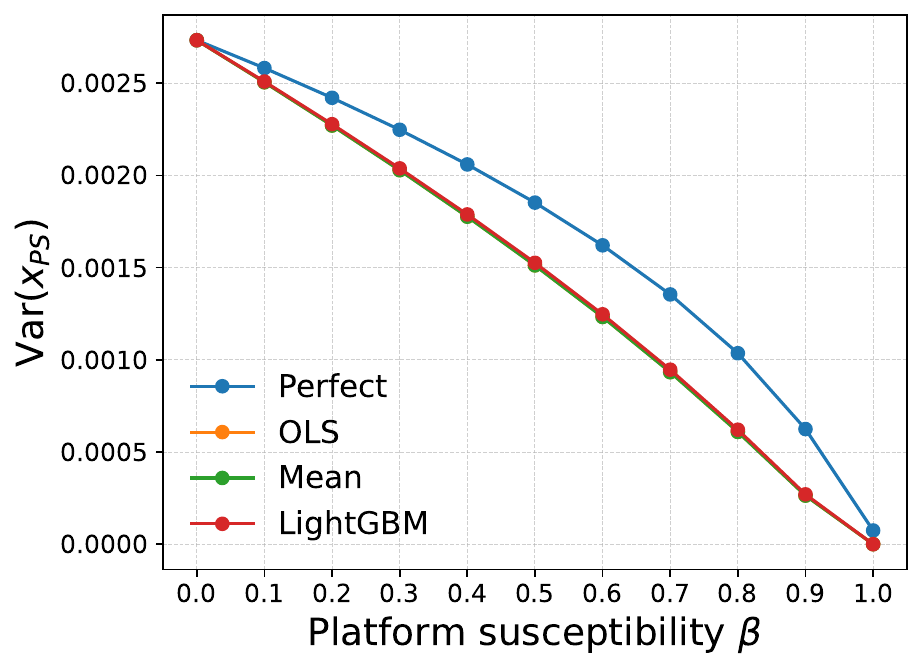}
            \caption{Opinions at equilibrium (Pokec)}
            \label{fig:platform_variance_sus_pokec}
        \end{subfigure}
    \caption{ In $(a)$, we show how performativity homogenizes opinions across retraining steps. The $x$-axis denotes the retraining step $t$, and the $y$-axis corresponds to individuals' expressed opinions in each times step, where the error bars indicate the variance. In (b), we show how platform susceptibility decreases variance of opinions at equilibrium. The $x$-axis denotes the varying homogeneous platform susceptibility while the $y$-axis denotes the variance of opinions in equilibrium. }
    \label{fig:consensus_sl_pokec} 
\end{figure*}

\begin{figure}[t]
    \begin{subfigure}[b]{0.45\linewidth}
        \centering
        \includegraphics[width=\linewidth]{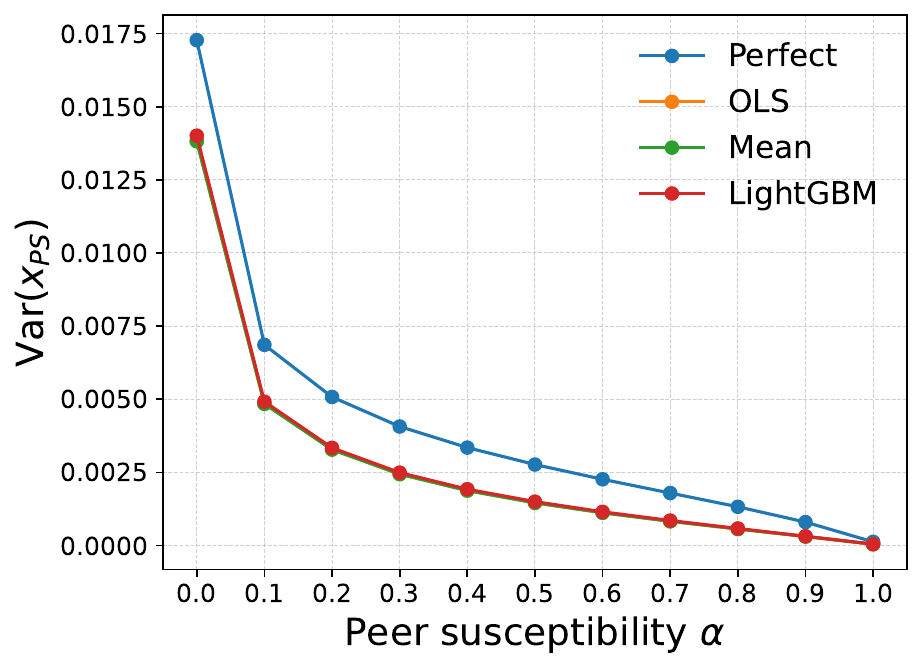}
        \caption{Opinions at equilibrium (Pokec) }
            \label{fig:peer_variance_sl_pokec}
    \end{subfigure}
    \hfill
    \begin{subfigure}[b]{0.45\linewidth}
    \centering 
    \includegraphics[width=\linewidth]{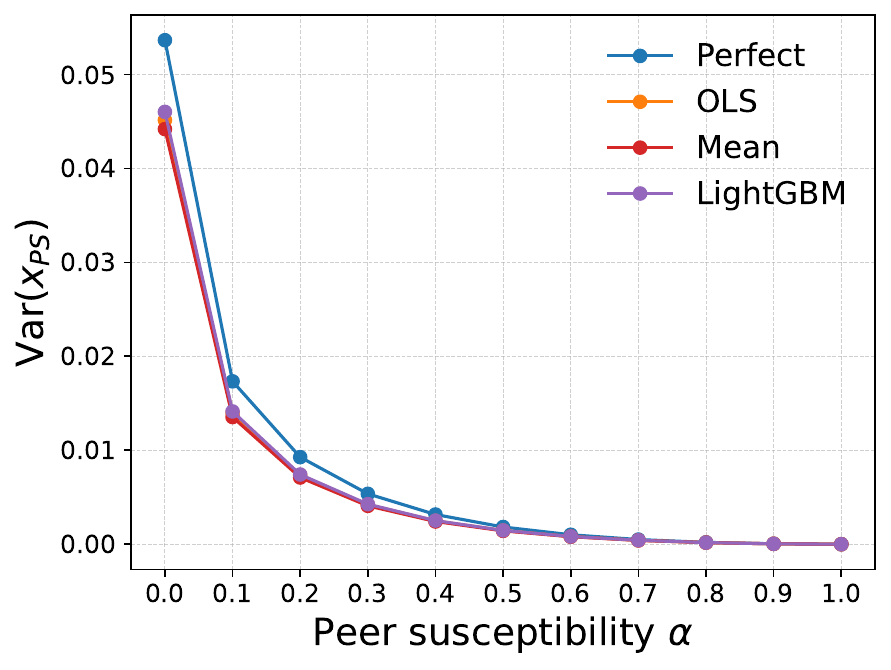}
    \caption{Opinions at equilibrium (Yelp)}
    \end{subfigure}
    \caption{ We show how peer susceptibility decreases variance at equilibrium. The $x$-axis denotes the varying homogeneous peer susceptibility while the $y$-axis denotes the variance of opinions in equilibrium. In~(a), we use Pokec dataset and in~(b), we use Yelp dataset.} 
    \label{fig:peer_variance_sl}
\end{figure}

\begin{figure}[t]
        \begin{subfigure}[b]{0.45\linewidth}
            \centering
            \includegraphics[width=\linewidth]{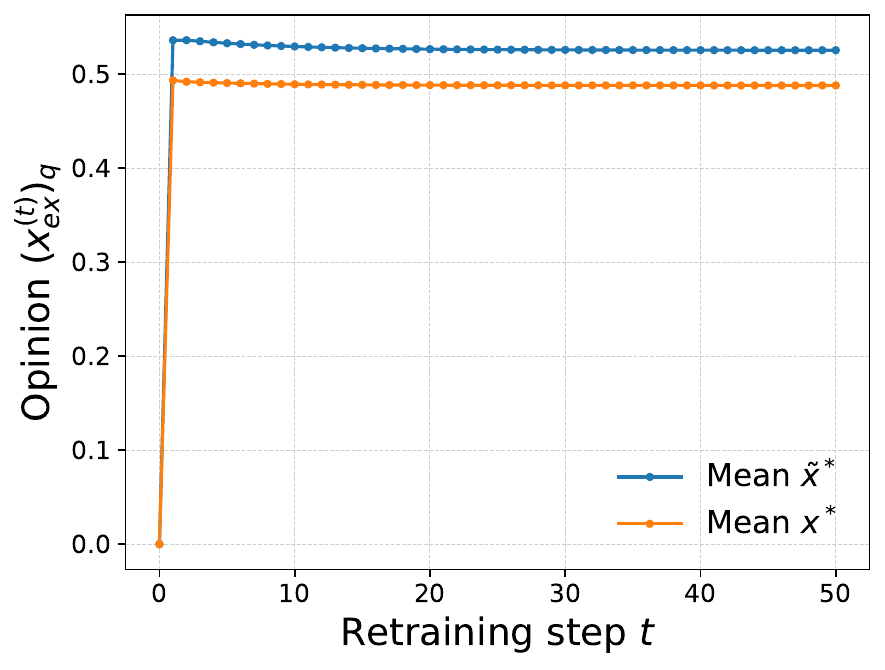 }
            \caption{Evolving opinion of stubborn individual (Pokec)
            }
        \end{subfigure}
        \hfill
        \begin{subfigure}[b]{0.45\linewidth}
            \centering
            \includegraphics[width=\linewidth]{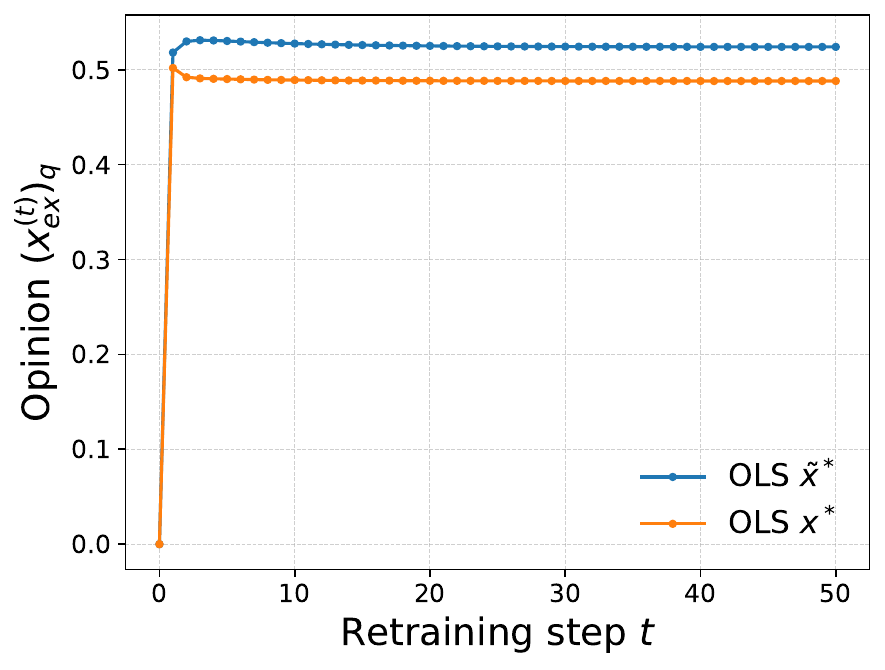}
            \caption{Evolving opinion of stubborn individual (Pokec)}
        \end{subfigure}
        \begin{subfigure}[b]{0.45\linewidth}
            \centering
            \includegraphics[width=\linewidth]{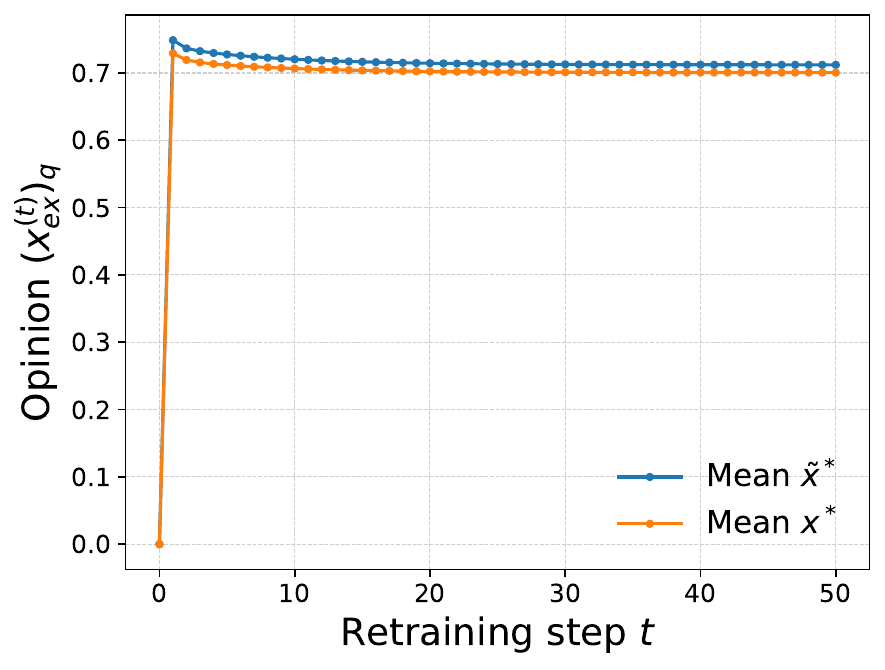 }
            \caption{Evolving opinion of stubborn individual (Yelp)
            }
        \end{subfigure}
        \hfill
        \begin{subfigure}[b]{0.45\linewidth}
            \centering
            \includegraphics[width=\linewidth]{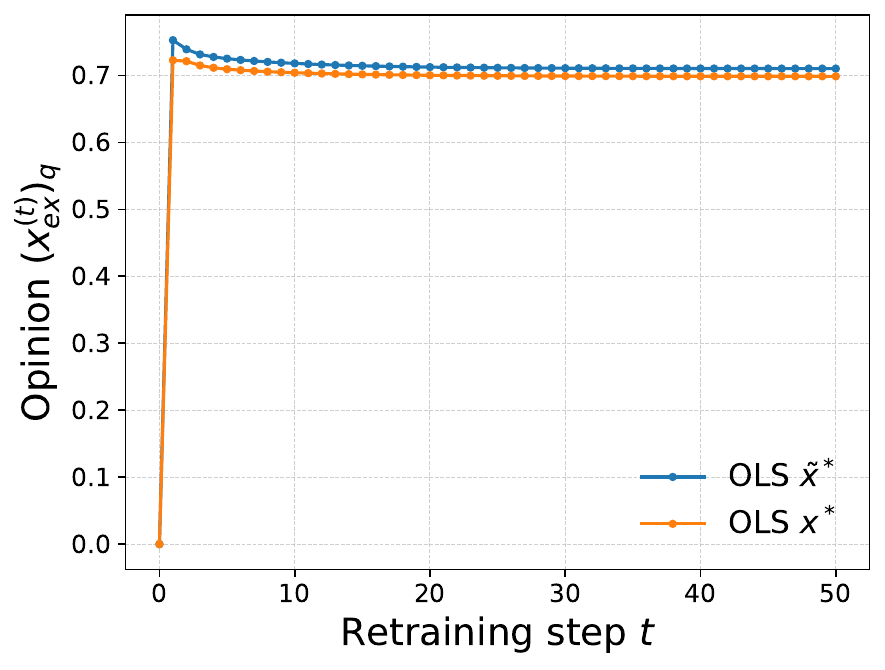}
            \caption{Evolving opinion of stubborn individual (Yelp)}
        \end{subfigure}
        \caption{ We show the individual $l\in U$ with $\alpha_l=0$ converges to different equilibria under different innate opinions of individuals in set $O$. The $x$-axis denotes the retraining steps, and the $y$-axis denotes the opinions of individual $l$. In~(a), the platform uses the mean estimation as in Proposition~\ref{prop:consensus_mean_estimation}. We use Pokec dataset for~(a) and (b), and Yelp for~(c) and (d). In~(b), the platform uses OLS for platform predictions. }
        \label{fig:steer2}
\end{figure}

\begin{figure}
        \begin{subfigure}[b]{0.45\linewidth}
            \centering
            \includegraphics[width=\linewidth]{figs/pokec_ridge_parametric_sl_retrain_steps_stubborn_unlabeled.pdf}
            \caption{Evolving opinion of stubborn individual (Pokec) }
            \label{fig:stubborn_pokec}
        \end{subfigure}
        \hfill
        \begin{subfigure}[b]{0.45\linewidth}
            \centering
            \includegraphics[width=\linewidth]{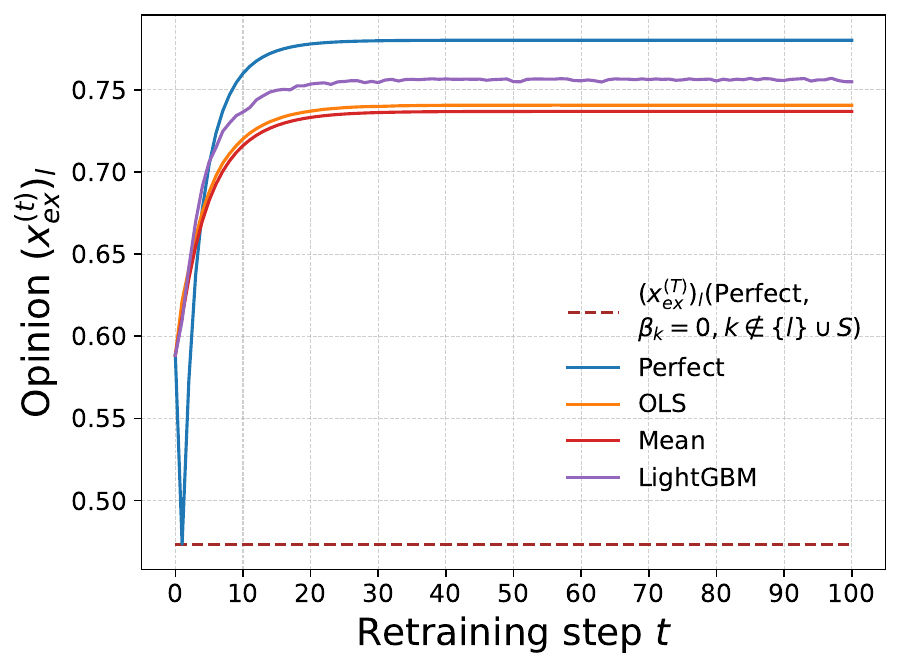}
            \caption{Indirect effect of platform intervention. (Pokec)}
            \label{fig:steer_pokec}
        \end{subfigure}
        \caption{ In $(a)$, we show how a stubborn individual $q\in U$ with $\alpha_q=0$ and $x^*_q=0$ is influenced by their peers across retraining steps for two different scenarios. In $(b)$, we show how indirect platform influence increases over retraining steps. We compare different the opinion of the stubborn individual $l$ with $\beta_l=0$. The purple dashed line denotes the opinion of the stubborn individual if we set $\beta_k=0$, $k\notin \{l\}\cup S$. }
        \label{fig:mean_steer_pokec}
\end{figure}

\end{document}